\documentclass[aps,pra,showpacs,superscriptaddress,twocolumn]{revtex4}
 \font\tenmsb=msbm10 scaled\magstep
1
\newfam\msbfam
   \textfont\msbfam=\tenmsb

\usepackage{hyperref}
\usepackage{bbm}
\usepackage{graphicx}
\usepackage{amsmath}
\usepackage{amssymb}
\usepackage{amsthm}
\usepackage{mathscinet}
\usepackage{amsfonts,amsmath}
\usepackage{color}
\usepackage{dashrule}
\newcommand{\id}{I}
\def\ket#1{| #1 \rangle}
\def\bra#1{\langle #1 |}
\def\kb#1#2{|#1\rangle\!\langle #2 |}
\def\bk#1#2{\langle #1 |#2\rangle}

\input{Qcircuit}

\newcommand{\tr}[1]{\text{Tr} #1}

\newcommand{\m}{\mathcal }

\newcommand{\diag}{{\rm diag}}   %%%%%%%%% diag
\def\C{\mathbb C}    %%%%%%%%% the set of complex numbers
\def\M{\mathbb M}    %%%%%%%%% matrices
\newtheorem{theorem}{Theorem}
\newtheorem{definition}[theorem]{Definition}
\newtheorem{corollary}[theorem]{Corollary}

\newtheorem{ex}[theorem]{Example} %*****NEW

\allowdisplaybreaks[1]%****To obtain page breaks in long multi-page equations
\begin{document}

\title[Exploring The Complementarity of Quantum Privacy and Error Correction]{Quantum Subsystems: Exploring the Complementarity of Quantum Privacy and Error Correction}

\author{Tomas Jochym-O'Connor}

\affiliation{Institute for Quantum Computing, University of Waterloo, Waterloo, Ontario, Canada}
\affiliation{Department of Physics \& Astronomy, University of Waterloo, Waterloo, Ontario, Canada}

\author{David W. Kribs}

\affiliation{Institute for Quantum Computing, University of Waterloo, Waterloo, Ontario, Canada}
\affiliation{Department of Mathematics \& Statistics, University of Guelph, Guelph, Ontario, Canada}

\author{Raymond Laflamme}

\affiliation{Institute for Quantum Computing, University of Waterloo, Waterloo, Ontario, Canada}
\affiliation{Department of Physics \& Astronomy, University of Waterloo, Waterloo, Ontario, Canada}
\affiliation{Perimeter Institute for Theoretical Physics, Waterloo, Ontario, Canada}

\author{Sarah Plosker}

\affiliation{Department of Mathematics \& Computer Science, Brandon University, Brandon, Manitoba, Canada}


\begin{abstract}
%This paper addresses and expands on the contents of the recent Letter \cite{JKLP12} discussing private quantum subsystems. Here we prove several results presented in \cite{JKLP12}, including a condition for a given random unitary channel to \emph{not} have a private subspace (although this does not mean that private communication cannot occur, as was demonstrated in \cite{JKLP12} via private subsystems) and  algebraic conditions that characterize when a general quantum subsystem or subspace code is private for a quantum channel. These conditions can be regarded as the private analogue of the Knill-Laflamme conditions for quantum error correction, and we explore how the conditions simplify in some special cases. The bridge between quantum cryptography and quantum error correction provided by complementary quantum channels motivates the study of a new, more general definition of quantum error correcting code as set out in \cite{JKLP12}, and we initiate this study here. We also consider the concept of complementarity for the general notion of private quantum subsystem. 
This paper addresses and expands on the contents of the recent Letter [\emph{Phys. Rev. Lett.} \textbf{111}, 030502 (2013)] discussing private quantum subsystems. Here we prove several previously presented results, including a condition for a given random unitary channel to \emph{not} have a private subspace (although this does not mean that private communication cannot occur, as was previously demonstrated via private subsystems) and algebraic conditions that characterize when a general quantum subsystem or subspace code is private for a quantum channel. These conditions can be regarded as the private analogue of the Knill-Laflamme conditions for quantum error correction, and we explore how the conditions simplify in some special cases. The bridge between quantum cryptography and quantum error correction provided by complementary quantum channels motivates the study of a new, more general definition of quantum error correcting code, and we initiate this study here. We also consider the concept of complementarity for the general notion of private quantum subsystem. 

\end{abstract}

% 03.67.Dd 	Quantum cryptography and communication security

\maketitle

\section{Introduction}

Quantum information processing and computing promises great advances in computational efficiency and the development of new cryptographic schemes. As such technologies continue to progress, techniques to control and manipulate quantum systems are at the forefront of research endeavours for both theoreticians and experimentalists. Simple concepts in classical computation, such as encoding a classical message using a private key and one-time pad, do not translate easily into the quantum setting. In this paper, we explore the algebraic and physical characteristics of privatizing quantum information, notably using subsystems of larger Hilbert spaces. Private quantum subsystems capture an important class of quantum cryptographic protocols as they take into account the ancillary space and allow for the communication of superpositions of states without requiring a full subspace structure. As such, understanding their underlying structure could have an important impact on the development of future cryptographic schemes.

Additionally, the known relationship between quantum privacy and error correction as complementary in the case of operator subsystem codes \cite{KKS08} suggests that progress in the understanding of quantum privacy could lead to further progress in quantum error correction, and vice-versa. In this work, we examine this notion of complementarity in more detail, expanding upon the results presented in~\cite{JKLP12}, and show that in the general setting the complementarity between privacy and error correction breaks down, but that the expected complementarity is recovered when one considers a certain larger Hilbert space and Stinespring dilation. Furthermore, inspired by the results on private quantum subsystems, a more general definition of quantum error correcting code was presented in~\cite{JKLP12}. Here we initiate an investigation of this generalized notion, and in particular we show how it is related to operator quantum error correction subsystem codes \cite{KLP05,KLPL06}.

This paper is structured as follows: In Section~\ref{sec:PrivateQuantumSubsystems} we review the different notions of quantum privacy, and summarize the results presented in~\cite{JKLP12} as well as provide the proofs of the main results presented therein on the contrast between private subspaces and subsystems. In Section~\ref{sec:TestableConditions} we review the algebraic conditions for private quantum subsystems, present the proof of the theorem presented in~\cite{JKLP12}, and illustrate the conditions in special cases and examples. Section~\ref{sec:Complementarity} discusses complementarity, showing how a straightforward generalization of the complementarity theorem of \cite{KKS08} fails, but that it is recovered on a larger Hilbert space. Section~\ref{sec:GenOQECC} revisits the general definition of quantum error correcting code presented in~\cite{JKLP12} and compares it to operator quantum error correction. We conclude in Section~\ref{sec:Conclusion} with possible applications, open questions and future directions of research on this subject.

\section{Private Quantum Subsystems In The Absence Of Private Subspaces}
\label{sec:PrivateQuantumSubsystems}

Here we consider the most general known notion of a private quantum code, which involves the encoding of quantum bits into subsystems. Given a Hilbert space $S$ representing our system, we call the Hilbert space $A$ (or $B$) a subsystem of $S$  if we can write $S$ as  $S=M \oplus M^\perp$, where $M$ is a subspace of $S$ that has a specified tensor decomposition as $M = A \otimes B$.  
Given a quantum system $S$, the sub\emph{spaces} of $S$ can
be viewed as subsystems $B$ for which $A$ is one-dimensional. A subscript such as $\sigma_B$ means the operator belongs to the class of linear operators of the subsystem $B$, denoted $\mathcal{L}(B)$. A quantum channel on $S$ is a completely positive trace-preserving map $\Phi$ on $\mathcal{L}(S)$.

The question of privacy for subsystems then becomes: Is there a subsystem $A$ that is private for $\Phi$? 
%Private quantum channels, private subspaces, and what we refer to as ``operator'' private subsystems---are captured as special cases of this general phenomena. ---cut & pasted from PRL so omit
Formally, we have the following definitions:

\begin{definition}
A subsystem $B$ is a \emph{private subsystem} for $\Phi$ if there is a $\rho_0 \in \mathcal{L}(S)$ and $\sigma_A \in \mathcal{L}(A)$ such that
\begin{eqnarray}\label{eq:private_definition}
\Phi (\sigma_A \otimes \sigma_B) = \rho_0 \quad \forall \sigma_B\in \mathcal{L}(B).
\end{eqnarray}

A subsystem $B$ is an \emph{operator private subsystem} for $\Phi$ if for all $\sigma_A \in \m L (A)$ there is a $\rho_0= \rho_0 (\sigma_A) \in \mathcal{L}(S)$ such that
\begin{eqnarray}\label{op_private_defn}
\Phi (\sigma_A \otimes \sigma_B) = \rho_0 \quad \forall \sigma_B\in \mathcal{L}(B).
\end{eqnarray}
\end{definition}
Operator private subsystems are precisely the private subsystems that are complementary to operator quantum error-correcting subsystems, as discussed in \cite{KKS08}.
If a channel $\Phi$ has an operator private (or operator error-correcting) subsystem, the map $\Phi$ becomes a product of channels on the individual subsystems $A$ and $B$ when
restricted to the combined product subspace $A \otimes B$. That is, since the output is independent of the inputs of both subsystems, the channel acts as $\Phi (\sigma_A \otimes \sigma_B)=\Phi_A(\sigma_A)\otimes \Phi_B(\sigma_B)\,\,\,  \forall \sigma_A\in \m L(A), \forall \sigma_B\in \mathcal{L}(B)$. Observe that such private subsystems cannot exist
without the existence of private subspaces; indeed, if equation ~(\ref{op_private_defn})
holds, it follows that every sub\emph{space}
$\ket{\psi}\otimes B$ is private for $\Phi$ for any fixed pure state $\ket{\psi}$ on $A$.

The general definition equation (\ref{eq:private_definition}) was introduced in \cite{AMTW00,BR03} under the moniker {\it private quantum channels} and called {\it completely private subsystems} in \cite{BRS04,BHS05}. However, previous examples that had appeared in the literature \cite{AMTW00,BR03,BRS04,BHS05} had either been of operator type, or were in fact already private subspaces. In \cite{JKLP12} we found the first example of a private subsystem that exists even in the absence of private subspaces (and so in particular is not of operator type). We present the full details of this analysis below. 

We began our investigation in~\cite{JKLP12} with a very basic channel: a channel with an equally weighted distribution of the identity~$I$ and Pauli~$Z$ as the Kraus operators of the channel, $\Lambda ( \rho) = (\rho + Z \rho Z)/2$. As argued in~\cite{JKLP12}, it is clear that no such channel can privatize a single qubit of information as it preserves all $Z$~basis information encoded into the state. However, could a multi-qubit version of such a channel; that is, a channel whose action is the above phase damping channel on each individual qubit $i$, denoted  $\Lambda = \Lambda_n \circ \hdots \circ \Lambda_1$, exhibit a private encoding of information? The answer to this question, with regard to private quantum subspaces, is given in the following general result.

\begin{theorem}\label{thm:RUCnoPS}
Let $\Phi$ be a random unitary channel with mutually commuting Kraus operators.
Then $\Phi$ has no private subspaces.
\label{lem:RandomUnitaryChannel}
\end{theorem}

\begin{proof}
Let $\Phi$ be a random unitary channel with mutually commuting Kraus operators %. Then the action of the channel on any input state $\rho \in \m{L}(\m{H})$ is 
described by
\begin{eqnarray}
\Phi (\rho) = \sum_i p_i U_i \rho U_i^{\dagger}\quad \forall \rho.
\end{eqnarray}
Since the  unitaries $ U_i $ are mutually commuting, there exists a common eigenbasis $\{\ket{e_j}\}_{j=1}^d$ for all of the unitaries such that,
\begin{eqnarray}
U_i \ket{e_j} = \alpha_{ij} \ket{e_j} \text{ with } \quad |\alpha_{ij}| = 1.
\end{eqnarray}
Suppose a non-trivial private subspace $C$ exists. Then there must exist at least two pure states $\ket{0_L}$, $\ket{1_L}$ such that $\Lambda (\ket{0_L}\bra{0_L}) = \Lambda( \ket{1_L}\bra{1_L}) = \rho_0$, where $\rho_0$ is some fixed density matrix. Then for some scalars $\beta_j, \gamma_j\in \C$, we can write
\begin{eqnarray}
\ket{0_L} = \sum_{j=1}^d \beta_j \ket{e_j}, \quad \ket{1_L} = \sum_{j=1}^d \gamma_j \ket{e_j}.
\end{eqnarray}
Consider the action of the channel on these  states:
\begin{eqnarray}
\Phi( \ket{0_L}\bra{0_L} ) &= \sum_i p_i U_i \Big( \sum_{j,k=1}^d \beta_j \beta_k^* \kb{e_j}{e_k} \Big) U_i^{\dagger} \nonumber \\
&= \sum_i p_i  \sum_{j,k=1}^d \alpha_{ij} \alpha_{ik}^* \beta_j \beta_k^* \kb{e_j}{e_k}  \nonumber \\
&=  \sum_{j,k=1}^d \Big( \sum_i p_i  \alpha_{ij} \alpha_{ik}^* \Big) \beta_j \beta_k^* \kb{e_j}{e_k}
\end{eqnarray}
Similarly,
\begin{eqnarray}
\Phi ( \ket{1_L}\bra{1_L} ) =  \sum_{j,k=1}^d \Big( \sum_i p_i  \alpha_{ij} \alpha_{ik}^* \Big) \gamma_j \gamma_k^* \ket{e_j}\bra{e_k}.
\end{eqnarray}
Comparing the diagonal terms, where $j=k$, the inside sum over $i$ is always equal to 1 since the modulus of the eigenvalues is 1, thus the respective coefficients are $|\beta_j|^2$ and $|\gamma_j|^2$. Therefore, if the  output of the channel is the same  in both cases, we must have  $|\beta_j| = |\gamma_j|$, $\forall j$. (Observe that this is independent of the orthogonality of $\ket{0_L}$ and $\ket{1_L}$; any two basis states mapped by $\Phi$ to the same state would satisfy this coefficient condition.)  However, we prove below that no such $\ket{0_L}$ and $\ket{1_L}$ can form a subspace. 
Indeed, we can write 
%Seems repetative: We proceed by contradiction: assume there is a subspace $\m{S}$ of dimension greater or equal to 2 that encodes a logical qubit and a state $\rho_0 \in \mathcal{H}$ such that $\forall \; \rho \in \m{L}(\mathcal{S})$, $\Lambda(\rho) = \rho_0$.
% Let us write $\ket{0_L}$ and $\ket{1_L}$  in terms of  the basis $\{\ket{e_j}\}$ given above,
\begin{eqnarray}
\ket{0_L} &= \sum_{j=1}^d\beta_j\ket{e_j}=|\beta_1|\ket{e_1}+\sum_{j=2}^d|\beta_j|e^{i\theta_j}\ket{e_j}\\
\ket{1_L} &= \sum_{j=1}^d\gamma_j\ket{e_j}=|\beta_1|\ket{e_1}+\sum_{j=2}^d|\beta_j|e^{i\phi_j}\ket{e_j},
\end{eqnarray}
where we have, without loss of generality, performed a global phase shift on the two vectors so that the coefficient of $\ket{e_1}$ is real \emph{for both vectors} (under a global phase shift, the vectors remain orthogonal). We have relabelled the coefficients to reflect the fact that  $|\beta_j|=|\gamma_j|$.
 
Any linear combination of the basis states must additionally be in $C$ by the closure of the subspace under scalar addition. With this in mind, consider the normalized state,
\begin{eqnarray}
\dfrac{\ket{0_L}+\ket{1_L}}{\sqrt{2}} &= \dfrac{2|\beta_1| \ket{e_1} +\sum_{j=2}^d|\beta_j|(e^{i \theta_j} + e^{i \phi_j})\ket{e_j}}{\sqrt{2}}.
\end{eqnarray}
Since such a state must be an element of $C$, it must satisfy the conditions on the moduli of its coefficients; namely, the $j$th coefficient must be equal in modulus to $|\beta_j|$. However, one can clearly see that the modulus of the coefficient of the $\ket{e_1}$ term is equal to $\sqrt{2} |\beta_1|$ which is not equal to $|\beta_1|$ unless $|\beta_1|=0$. Therefore, we have reduced the basis states to have the form,
\begin{eqnarray}
\ket{0_L} &=  |\beta_2|\ket{e_2}+\sum_{j=3}^d|\beta_j|e^{i\theta'_j}\ket{e_j} \\
\ket{1_L} &=  |\beta_2|\ket{e_2}+\sum_{j=3}^d|\beta_j|e^{i\phi'_j}\ket{e_j} ,
\end{eqnarray}
where we have performed a global phase shift on both states and redefined the phase on the components $\ket{e_j}, 3\leq j\leq d$. By the same argument as above, we can show that all coefficients must be equal to zero in order for the channel $\Lambda$ to be private while $C$ remains a subspace. As such, there does not exist two orthonormal basis vectors satisfying the requirements for the channel to be private, implying that no non-trivial subspace~$C \subset S$ exists. 

\end{proof}

\begin{corollary}
Let $S$ be $n$-qubit Hilbert space. Then there exists no subspace~$C \subset S$ where $dim(C) \ge 2$ such that $C$ is private for the channel $\Lambda=\Lambda_n\circ\Lambda_{n-1}\circ\cdots\circ \Lambda_1$.
\end{corollary}

\begin{proof}
All Kraus operators are tensor products of $I_2$ and $Z$, and thus Theorem \ref{thm:RUCnoPS} applies. 
\end{proof}

%I find this paragraph confusing: In conclusion, we cannot find a subspace encoding whose single qubit of information is stored in the off-diagonal elements of the quantum state, written in the standard computational basis. However, one can find a two-qubit subsystem encoding of information who's quantum information is stored in the off-diagonal components of the matrix. This highlights the difference between private quantum subspaces and subsystems.

For our phase-damping channel $\Lambda$, it is impossible to find a non-trivial subspace of encoded  qubits that can be completely stored in their off-diagonal elements (recall that $\Lambda$ fixes the diagonal elements of any input state). However, as we show below, one can find a single-qubit quantum subsystem code, even in the two-qubit case $n=2$, that is private for $\Lambda$. %, which encodes their information  in off-diagonal elements. 
This highlights the difference between private quantum subspaces and subsystems. We provide details of this encoding presently.

Consider the following logically encoded qubits in two-qubit Hilbert space:
\begin{eqnarray}\label{subsystem_encoding}
\rho_L=\frac14(II+\alpha XX + \beta YI+\gamma ZX). 
\end{eqnarray}
This describes a single qubit encoding, as equation ~(\ref{subsystem_encoding}) describes the coordinates for a logical Bloch sphere in two-qubit Hilbert space with logical Pauli operators given by $X_L=XX, Y_L=YI, Z_L=ZX$. Now, observe that the dephasing map $\Lambda = \Lambda_2 \circ \Lambda_1$ acting on each density operator $\rho_L$ produces an output state that is maximally mixed; that is, $\Lambda(\rho_L) = \frac14 \, II$ for all $\rho_L$. Thus, we see that equation ~(\ref{subsystem_encoding}) yields a private single-qubit code for the
dephasing map $\Lambda$. 

 We claim that this private code can be viewed as a single qubit subsystem embedded inside two qubit space, where the ancilla operator $\sigma_A$, from equation~(\ref{eq:private_definition}), in this case is the single qubit identity operator $I_2$; that is, up to a unitary equivalence the set of operators $\rho_L$ can be seen to generate the operator algebra $I_2 \otimes \mathbb{M}_2$. To see this, it is enough to show that all two-qubit states $\rho$ of the form $\rho~=~\frac14(II~+~\alpha XX~+~\beta YI+\gamma ZX)$ can be sent through appropriate unitary gates to obtain $\rho'$ of the form $\rho'~=~I_2\otimes\frac14(I_2~+~\alpha'X~+~\beta' Y~+~\gamma' Z) $. Since $I, X, Y, Z$ form a basis for $\mathbb{M}_2$, the claim will follow.

We find that an application of the inverse of the $K$-gate,
\[
K=\frac1{\sqrt{2}} \begin{pmatrix}
1 & 1 \\
i & -i
\end{pmatrix}=\frac1{\sqrt{2}} (\ket{0}(\bra{0} + \bra{1}) + i \ket{1}(\bra{0} - \bra{1})),
\]
on the first qubit, and applications of $CNOT_{2,1}$ and $CNOT_{1,2}$, yields the desired transformation. 
Indeed, the composition $CNOT_{1,2} CNOT_{2,1} \big((K^\dagger \otimes I_2)(\cdot)(K\otimes I_2)\big)CNOT_{2,1}CNOT_{1,2}$ acts as:
\begin{eqnarray*}
XX \longmapsto	& YX &\longmapsto	 ZY \longmapsto IY\\
YI \longmapsto     	& ZI  &\longmapsto	 ZZ \longmapsto IZ\\
ZX \longmapsto    	& XX &\longmapsto	 IX \longmapsto IX.
\end{eqnarray*}
Thus, we obtain  $\rho'=\frac14(I_4+\gamma IX+\alpha IY+\beta IZ)$. In particular, by defining the unitary 
\[
U = CNOT_{1,2}\circ CNOT_{2,1}\circ(K^\dagger\otimes I_2)
=\frac{1}{\sqrt{2}}\begin{pmatrix}
1 & 0 & -i& 0\\
0 & 1 & 0 & i\\
0 & 1 & 0 & -i\\
1 & 0 & i & 0
\end{pmatrix}, 
\]
we see the set of operators $U\rho_L U^\dagger$ generate the algebra $I_2 \otimes \mathbb{M}_2$.

Thus, this subsystem encoding fits into the framework of the definition of private quantum subsystem; that is, the subsystem defined by the set of operators $U\rho_L U^\dagger$ is a private subsystem for the channel $\Lambda$. In fact, it is a private subsystem that is not operator private. This follows from Theorem~\ref{lem:RandomUnitaryChannel} together with the complementarity theorem of \cite{KKS08}: If the subsystem was operator private, it would complement an operator quantum error correcting code (discussed in the next section), which would imply the complementary channel also has a correctable subspace code of the same size, and then incorrectly imply that the original channel has a private subspace again by complementarity. For completeness, we will show directly below that this private subsystem is not operator private.

We have shown explicitly that $\m L(C)$  is isomorphic to $ I_2\otimes \M_2$, where $C$ is the set of all private states for $\Lambda$, 
 via $U\m L(C)U^\dagger =I_2\otimes \M_2$. Alternatively, we can consider the modified channel $\Lambda'(\cdot) := U \Lambda (\cdot) U'$. Then the second qubit in the standard computational basis decomposition $A \otimes B$, $A = \C^2 = B$, is private for $\Lambda'$ with $\sigma_A = \frac{1}{2} I_2$. That is, rather than applying the unitary transformation $\rho_L\mapsto U\rho_LU^\dagger$ and sending this resulting state through the channel $\Lambda$, we can modify the Kraus operators of $\Lambda$ by the same unitary $U$ so that $\Lambda'$ is private for $\frac12I_2\otimes \sigma_B$ for any $\sigma_B\in\m L(B)$. In this manner, our example directly fits the definition of private subsystem.

 For any
 \[
 \sigma_A=\begin{pmatrix}
 a_A & b_A \\
 c_A & d_A
 \end{pmatrix}\in \m L(A),\quad
\sigma_B=\begin{pmatrix}
 a_B & b_B \\
 c_B & d_B
 \end{pmatrix}\in \m L(B),
 \]
  we compute
\begin{eqnarray}\label{eq:notpoutput}
\Lambda'(\sigma_A\otimes \sigma_B)= 
\dfrac12
  \begin{pmatrix}
\gamma_{AB} & 0 & 0 & \eta_{AB} \\
 0 & \gamma_{AB}  & \zeta_{AB} & 0\\
 0 & \zeta_{AB} & \gamma_{AB} & 0\\
\eta_{AB} & 0 & 0 & \gamma_{AB} 
 \end{pmatrix}.
%  \begin{pmatrix}
% a_Aa_B+d_Ad_B& 0 & 0 & b_Ab_B+c_Ac_B\\
% 0 & a_Aa_B+d_Ad_B & b_Ac_B+c_Ab_B& 0\\
% 0 &   b_Ac_B+c_Ab_B & a_Aa_B+d_Ad_B& 0\\
% b_Ab_B+c_Ac_B& 0 & 0 & a_Aa_B+d_Ad_B
% \end{pmatrix}.
 \end{eqnarray}
 where $\gamma_{AB} = a_Aa_B + d_A d_B$, $\eta_{AB} = b_A b_B + c_A c_B$, and $\zeta_{AB} = b_A c_B +c_A b_B$. Note that this output is symmetric in the subsystems $A$ and $B$. In particular, $\Lambda'(\sigma_A\otimes \frac12 I_2)=\frac14\diag(a_A+d_A, a_A+d_A, a_A+d_A, a_A+d_A)$ and $\Lambda'(\frac12 I_2\otimes\sigma_B)=\frac14\diag(a_B+d_B, a_B+d_B, a_B+d_B, a_A+d_B)$. Thus for density matrices $\sigma_A$, $\sigma_B$ we have $\Lambda'(\sigma_A\otimes \frac12I_2)=\Lambda'(\frac12 I_2\otimes\sigma_B)=\frac14 I_4$, and so both the first and second computational basis subsystems are private for $\Lambda'$.\\
 
If we were looking at an operator private subsystem here, the channel $\Lambda'$ would split up into two distinct channels acting on systems $A$ and $B$ respectively. Thus, we would have density matrices $\tau_A$, $\tau_B$ such that  $\Lambda'( \sigma_{A}\otimes  \sigma_{B})=\tau_A\otimes\tau_B$.
 Equating this equation with equation (\ref{eq:notpoutput}), we find the system has no solution for general $\sigma_A, \sigma_B$ (equating components forces $\tau_B$ to be the zero matrix, which then forces $\Lambda'(\sigma_A\otimes \sigma_B)$ to be the zero matrix). Hence, this subsystem is private for $\Lambda'$, but not operator private.

\section{Testable Conditions For Private Quantum Codes}
\label{sec:TestableConditions}

%These two paragraphs copy-pasted from PRL, so omitted
%If we are given a quantum channel $\Phi(\rho) = \sum_i V_i \rho V_i^\dagger$ and a subsystem $B$, we can ask if it is possible to decide whether $B$ is private for $\Phi$; and more to the point, we can ask if this can be answered in terms of the Kraus operators $V_i$ for the channel. These conditions can be thought of as Knill-Laflamme conditions in that they are a set of algebraic constraints giving necessary and sufficient conditions for the existence of private quantum subsystems. 

%The following result answers this question for private quantum subsystems. As one might expect, these conditions are based on not only the Kraus operators of the channel, but also the eigenvalues and eigenvectors of the fixed $A$ state $\sigma_A$ and the output state $\rho_0$. 

The following theorem was presented in \cite{JKLP12} and gives algebraic conditions on the Kraus operators of a channel that are necessary and sufficient for the existence of a private subsystem. These conditions necessarily involve the eigenvalues and eigenvectors of the fixed $A$ state $\sigma_A$ and the output state $\rho_0$. We prove the result below. 

\begin{theorem}\label{thm:KLsubsyst}
 A subsystem $B$ is private for a channel
$\Phi(\rho) = \sum_i V_i \rho V_i^\dagger$ with fixed $A$
state $\sigma_A$ and output state $\rho_0$ if and only if
there are complex scalars $\lambda_{ijkl}$ forming an isometry matrix $\lambda = (\lambda_{ijkl})$, and 
\begin{eqnarray}\label{testablecondition}
&\sum_m \sqrt{p_k }V_j\ket{\psi_{A,k}} \ket{\psi_{B,m}}  \bra{\psi_{B,m}} \qquad\\
& \qquad \qquad \qquad = \sum_{i,l}\lambda_{ijkl} \sqrt{q_l}\ket{\phi_{l}}\bra{\psi_{B,i}}       
\qquad  \forall \text{ }j, \text{ } k, \nonumber 
\end{eqnarray}
viewing both sides as a matrix representing a linear map preserving trace-distance between operators, where $\ket{\psi_{A,k}}$ ($p_k$) and $\ket{\phi_{l}}$ ($q_l$) are eigenstates (eigenvalues) of
$\sigma_A$ and $\rho_0$ respectively, and $\ket{\psi_{B,i}}$ is an orthonormal basis
for $B$. %, and where $\ket{\psi_{A,k}}$ is viewed as a channel from $B$ into $S$.
\end{theorem}

%Note: We regard $\ket{\psi_{A,k}}$ as operators from the $A$ subsystem to the full tensor product  $A\otimes B$ subspace. That is,  $\ket{\psi_{A,k}}:\ket{\psi_B}\mapsto\ket{\psi_{A,k}}\otimes \ket{\psi_B}$ for all states $\ket{\psi_B}$ in the subsystem $B$. In this way, $V_j \ket{\psi_{A,k}}$ is well-defined.\\

\begin{proof}
 Consider first the left-hand side of the equation (\ref{eq:private_definition}) of  the definition of private quantum subsystem. Let $\Phi:\mathcal{L}(S)\rightarrow\mathcal{L}(S)$ be a quantum channel satisfying this definition. Let $\{V_j\}$ be the Kraus operators of $\Phi$. Consider a spectral decomposition $\sigma_A=\sum_kp_k\ket{\psi_{A,k}}\bra{\psi_{A,k}}$, where $\ket{\psi_{A,k}}$ and $p_k$ are the eigenstates and eigenvalues, respectively, of $\sigma_A$. We can consider the action of $\Phi$ on $\m L(S)$ as the composition of maps $\Phi\circ\Psi(\sigma_B)$, where, for fixed $\sigma_A$, $\Psi:\mathcal{L}(B)\rightarrow\mathcal{L}(S)$ is the map $\sigma_B\mapsto\sigma_A\otimes\sigma_B$. The Kraus operators of $\Psi$ are $\{\sum_m\sqrt{p_k}\ket{\psi_{A,k}}\ket{\psi_{B, m}}\bra{\psi_{B, m}}\}_k$ (the $\sum_m \ket{\psi_{B, m}}\bra{\psi_{B, m}}$ acts trivially on $\sigma_B$, but is necessary to obtain the correct dimension when later acted on by $V_j$).  It follows that the Kraus operators of the composition $\Phi\circ\Psi(\sigma_B)$ are $\{\sum_m\sqrt{p_k}V_j \ket{\psi_{A,k}}\ket{\psi_{B, m}}\bra{\psi_{B, m}}\}_{j,k}$.

On the other hand, the right-hand side of equation (\ref{eq:private_definition}) can be viewed as a quantum channel \[
\sigma_B\mapsto\tr(\sigma_B)\sum_l  q_l \ket{\phi_l }\bra{\phi_l }=\sum_{i,l }q_l \ket{\phi_l }\bra{\psi_{B,i}}\sigma_B\ket{\psi_{B,i}}\bra{\phi_l },\]
 where $\{\ket{\psi_{B,i}}\}$ is an orthonormal basis for the subsystem $B$, and  we have used the fact that $\ket{\phi_l }$ and $q_l $ form a spectral decomposition for $\rho_0$. The Kraus operators of this map are $\{\sqrt{q_l }\ket{\phi_l }\bra{\psi_{B,i}}\}_{i,l}$.

However, the quantum channels described by the left- and right-hand sides of equation (\ref{eq:private_definition}) are equal in that, given an arbitrary input $\sigma_B$, their outputs are equal.  Thus we may use a well-known fact regarding equal CP maps with Kraus operators $\{X_i\}_{i=1}^m$ and $\{Y_j\}_{j=1}^n$, respectively, with $ m\leq n$; that is, they are related via
$X_i=\sum_j\lambda_{ij}Y_j$ for some isometry matrix $\lambda=(\lambda_{ij})$. When $m=n$, $\lambda$ is unitary. It follows immediately from this that for all $j,k$, we have $\sum_m \sqrt{p_k }V_j\ket{\psi_{A,k}} \ket{\psi_{B,m}}\bra{\psi_{B,m}} = \sum_{i,l }\lambda_{ijkl} \sqrt{q_l }\ket{\phi_{l}}\bra{\psi_{B,i}}$,  for some isometry (or, appropriately, unitary) $\lambda$, as desired. 

The converse implication follows by reversing the above steps, or by direct calculation, to show that equation~(\ref{testablecondition}) implies equation (\ref{eq:private_definition}) is satisfied. 
\end{proof}

The conditions of Theorem~\ref{thm:KLsubsyst} are somewhat intricate in the most general case, so it is worthwhile to give further context and discuss some special cases. We note that this result is new even for the special cases of operator private codes and private subspaces, and, via complementarity, the result can thus be viewed as the quantum privacy analogue of the Knill-Laflamme theorem for quantum error-correcting (subspace) codes \cite{KL97} and its operator quantum error correction generalization \cite{KLP05, KLPL06}. However, the most general case covered by Theorem~\ref{thm:KLsubsyst} may have no analogue in quantum error correction. The next two sections discuss this topic in more detail. 

As one would expect, the algebraic conditions can be further simplified in the case of private subspaces; which is captured in the formalism when $A$ is one-dimensional and $B$ is a subspace. In this case, the Theorem statement becomes $V_j P_B = \sum_{i,l} \lambda_{ijl} \sqrt{q_l} \ket{\phi_l}\bra{\psi_{B,i}}$. By taking the inner product of this equation with its complex conjugate, one arrives at the statement $P_B V_{j_1}^\dagger V_{j_2} P_B=\sum_{i_1,i_2,l} q_l \overline{\lambda_{ij_1l}} \lambda_{ij_2l} \ket{\psi_{B,i_1}}\bra{\psi_{B,i_2}}$ for all $j_1$, $j_2$, where $P_B$ is the projector onto the $B$~subspace. %Hence a more direct connection with the Knill-Laflamme conditions for quantum error correction is evident in this case, which state $P_B V_{j_1}^\dagger V_{j_2} P_B = \alpha_{j_1,j_2} P_B$, where here $V_j$ would represent the Kraus operators for the error map and $P_B$ would be the quantum code.
Here we have a more noticeable connection with the Knill-Laflamme conditions for quantum error correction: $P_B V_{j_1}^\dagger V_{j_2} P_B = \alpha_{j_1,j_2} P_B$, where the $V_j$'s are the Kraus operators of the \emph{error  map} and $P_B$ is the projection onto the \emph{correctable} $B$ subspace.

The algebraic conditions of the theorem can also be simplified in the case that $\rho_0$ is a scalar multiple of a projection, as we now state. 

\begin{corollary}
Suppose the output state $\rho_0$ of a private quantum channel $\Phi=\{V_i\}$ is proportional to a projection: $\rho_0 \propto Q = \sum_k \ket{\psi_k}\bra{\psi_k}$, and $P = \sum_l \ket{\phi_l}\bra{\phi_l}$. It follows that there are scalars $u_{ikl}$ such that for all $i$
                \[
                V_i P = \sum_{k,l} u_{ikl} A_{kl} \quad{\rm where}\quad A_{kl} = \frac{1}{\sqrt{rank(Q)}} \ket{\psi_k}\bra{\phi_l}.
                \]
Here  $A_{kl}$ are the Kraus operators of the channel $X\mapsto\lambda_XP$. 
\end{corollary}
%
%It is important to note that this result is new even for private sub\emph{spaces}. In the notation of the theorem for that case, $A$ is one-dimensional and $B$ is the subspace. If we let $P_B$ be the projector of $S$ onto $B$, then we see that the characterization of privacy is given by the conditions: $V_j P_B = \sum_{i,l}\lambda_{ijl} \sqrt{q_l}\ket{\phi_{l}}\bra{\psi_{B,i}}$ for all $j$. As a simple illustration, in the case of the completely depolarizing channel on $N$-dimensional Hilbert space, $P_B$ is the identity operator and these conditions reduce to the Kraus operators satisfying $\sqrt{N}\, V_j = \sum_{i_1,i_2}  \lambda_{i_1i_2 j}  \ket{i_1}\bra{i_2}$, for some choice of orthonormal bases $\ket{i_1}$ and $\ket{i_2}$ and unitary matrix $(\lambda_{i_1i_2 j})_{i_1,i_2}$. 

 %In that case both $A$ and $B$ are spanned by $\{\ket{0}$, $\ket{1}\}$. The eigenstates of $\rho_0=\frac14 I_4$ are $\{\ket{00}, \ket{01}, \ket{10}, \ket{11}\}$, each having eigenvalue $\frac14$. The eigenstates of $\sigma_A=\frac12 I_2$ are $\{\ket{\psi_{A,k}}\}=\{\ket{0}, \ket{1}\}$, with corresponding eigenvalues $\frac12$. 

Thus far in our investigations, most of the physical examples of private codes that we have come across do indeed have a projector output as in this Corollary. Of course, the simplest general class of channels satisfying this condition is the $n$-qubit complete depolarizing channel. In that case, both $P$ and $Q$ are the maximally mixed state, and the result indicates that any family of Kraus operators for the map will arise as linear combinations, where the scalars are precisely defined with the right balance to induce privacy, of the matrix units $\ket{i}\bra{j}$. Another simple (non-unital) example is provided by the spontaneous emission channel. In the single qubit case, the extremal channel from this class is given by $\Phi(\rho) = \ket{0}\bra{0}$ for all single qubit $\rho$. Here $P$ is the maximally mixed state and $Q= \ket{0}\bra{0}$, and the result simply states that any Kraus operators for $\Phi$ must be balanced multiples of  $\ket{0}\bra{0}$ and $\ket{0}\bra{1}$. 

As a more intricate example in the most general (non-subspace, non-operator) case of a private code, we point out how the 2-qubit phase damping channel $\Lambda$ can be viewed from the perspective of this result. The eigenstates of $\rho_0=\frac14 I_4$ are $\ket{00}, \ket{01}, \ket{10}, \ket{11}$, each having eigenvalue $\frac14$. For simplicity, we will use the standard orthonormal basis on the subsystem $B$: $\{\bra{\psi_{B,i}}\}=\{\bra{0},\bra{1}\}$. %Then  the $4\times 2$ matrices $\ket{\phi_{l}}\bra{\psi_{B,i}}$ will have 1 in the $(\ell, i)$-th component, and zero elsewhere.
In our example, $\sigma_A=\frac12 I_2$, hence its eigenstates are $\{\ket{\psi_{A,k}}\}=\{\ket{0}, \ket{1}\}$, with corresponding eigenvalues $\frac12$.

Using the Kraus operators $\{V_j\}=\{\frac{1}{2} II, \frac{1}{2} XX, \frac{1}{2} ZZ,  -\frac{1}{2} YY\}$ of $\Lambda'$, we compute $V_j \ket{\psi_{A,k}}$ as follows:
\begin{eqnarray*}
V_1 \ket{\psi_{A,1}}&=&\frac1{2}\begin{pmatrix} 1& 0 \\ 0& 1\\ 0&0\\0&0\end{pmatrix}=\frac12\left(\ket{00}\bra{0}+\ket{01}\bra{1}\right)\\
V_1 \ket{\psi_{A,2}}&=&\frac1{2}\begin{pmatrix}  0&0\\0&0\\ 1& 0 \\ 0& 1\end{pmatrix}=\frac1{2}\left(\ket{10}\bra{0}+\ket{11}\bra{1}\right)\\
V_2 \ket{\psi_{A,1}}&=&\frac1{2}\begin{pmatrix} 0& 0 \\ 0& 0\\ 0&1\\ 1&0\end{pmatrix}=\frac1{2}\left(\ket{10}\bra{1}+\ket{11}\bra{0}\right)\\ 
V_2 \ket{\psi_{A,2}}&=&\frac1{2}\begin{pmatrix}  0&1\\ 1&0\\0& 0 \\ 0& 0\end{pmatrix}=\frac1{2}\left(\ket{00}\bra{1}+\ket{01}\bra{0}\right)\\
V_3 \ket{\psi_{A,1}}&=&\frac12\begin{pmatrix} 1& 0 \\ 0& -1\\ 0&0\\0&0\end{pmatrix}=\frac1{2}\left(\ket{00}\bra{0}-\ket{01}\bra{1}\right)\\ 
V_3 \ket{\psi_{A,2}}&=&\frac12\begin{pmatrix}  0&0\\0&0\\-1& 0 \\ 0& 1\end{pmatrix}=\frac1{2}\left(-\ket{10}\bra{0}+\ket{11}\bra{1}\right)\\
V_4 \ket{\psi_{A,1}}&=&\frac12\begin{pmatrix} 0&0\\0&0\\0& -1 \\ 1& 0\end{pmatrix}=\frac1{2}\left(-\ket{10}\bra{1}+\ket{11}\bra{0}\right)\\ 
V_4 \ket{\psi_{A,2}}&=&\frac12\begin{pmatrix}  0& 1 \\ -1& 0\\ 0&0\\0&0\end{pmatrix}=\frac1{2}\left(\ket{00}\bra{1}+\ket{01}\bra{0}\right).
\end{eqnarray*}

Note that the $V_j$  are $4\times 2$ matrices formed with $2\times 2$ Pauli operators and zero blocks. 
Recall that we can consider both the left-hand and right-hand side of equation~(\ref{testablecondition})
as quantum channels. Moreover, the Kraus operators $\{X_i\}, \{Y_j\}$ of equal quantum channels are related via
$X_i=\sum_j\lambda_{ij}Y_j$ for some isometry $\lambda=(\lambda_{ij})$. When  the number of Kraus operators $X_i$ is equal to the number of Kraus operators $Y_i$, $\lambda$ is unitary. In this case, $\sqrt{p_k}=\frac1{\sqrt{2}}$ for all $k$, and each $V_j$ has a factor of $\frac12$, so the coefficient of the left-side of this equation is always $\frac1{2\sqrt{2}}$. The coefficient of $\lambda_{ijkl} \sqrt{q_l }\ket{\phi_{l}}\bra{\psi_{B,i}}$ is $\lambda_{ijkl} \sqrt{q_l }=\frac1{\sqrt{2}}\cdot \frac1{2}$ for all $i,j,k,l $.

Thus in our example, we   find that $\lambda$ is the following matrix:
\[
\lambda=\frac1{\sqrt{2}}\begin{pmatrix} 1& 0 & 0 & 1&0 & 0 & 0&0 \\0 & 0 & 0& 0& 1&0 & 0 & 1\\0 & 0 & 0&0&0&1 & 1& 0\\0&1& 1 & 0 & 0& 0 & 0 & 0\\ 1 & 0 & 0&-1&0 & 0 & 0&0\\0 & 0 & 0&0 & -1 & 0&0&1\\0 & 0 & 0&0 & 0 & -1&1&0\\0&1&-1&0 & 0 & 0 & 0 & 0\end{pmatrix}.
\]
The scalar matrix $\lambda=(\lambda_{ijkl})$ is indeed an isometry. Furthermore, because the number of  operators $\sum_m\sqrt{p_k }V_j\ket{\psi_{A,k}} \ket{\psi_{B,m}}\bra{\psi_{B,m}}$ agrees with the number of operators $\ket{\phi_{l}}\bra{\psi_{B,i}}$ (namely, 8), the matrix $\lambda$ is in fact unitary.

\section{Extension of Complementarity}
\label{sec:Complementarity}

\subsection{Connection to quantum error correction}

The Stinespring dilation theorem \cite{Sti55}, gives the standard operational description of a quantum channel: Every channel $\Phi$ on a Hilbert space $S$ can be described by an environment Hilbert space $E$, a pure state $|\psi \rangle $ on $E$, and a unitary operator $U$ on the composite $%
S\otimes E $ as follows: $\Phi(\rho )=\tr_{E}\big(%
U(\rho \otimes \ket{\psi}\bra{\psi})U^\dagger\big).$ Tracing out the system $S$ instead yields a complementary channel: $\Phi^{\sharp}(\rho )=\tr_{S}\big(U(\rho \otimes
\ket{\psi}\bra{\psi})U^\dagger\big).$ The uniqueness (up to conjugation by a partial isometry) of the Stinespring dilation allows us to talk of ``the'' complementary channel $\Phi^\sharp$ for a given channel $\Phi$ \cite{Hol06,KMNR07}.

The complementarity theorem of \cite{KKS08} shows that a subsystem code is operator quantum error-correcting for a channel if and only if it is operator private for the corresponding complementary channel. One can then ask if this complementarity theorem extends to the setting of general private quantum subsystems and some more general notion of quantum error correcting code. We show this is not the case in the following discussion, which focusses on the class of phase damping examples considered above. However, in the general discussion that follows, we show how a modified view of the associated dilations recaptures the complementarity result. 

For our phase damping channel $\Lambda$, we can compute the Kraus operators of the complementary channel $\Lambda^\sharp$ by ``stacking'' the $j$-th column of each of the eight Kraus operators $V_i$ of $\Lambda$ one below the next, to obtain the $j$-th Kraus operator of  $\Lambda^\sharp $: 
\begin{eqnarray*}
A_1&=&\frac12\begin{pmatrix}
1 &0 & 0 & 0\\
1 &0 & 0 & 0\\
1 &0 & 0 & 0\\
1 &0 & 0 & 0\\
\end{pmatrix}
\quad
A_2=\frac12\begin{pmatrix}
 0 & 1 & 0 & 0\\
 0 & -1& 0 & 0\\
 0 & 1& 0 & 0\\
 0 & -1& 0 & 0\\
\end{pmatrix}
\\ 
A_3&=&\frac12\begin{pmatrix}
0 & 0 & 1 & 0\\
0 & 0 & 1 & 0\\
0 & 0 &-1& 0\\
0 & 0 &-1& 0\\
\end{pmatrix}
\quad
A_4=\frac12\begin{pmatrix}
 0 &0 & 0& 1\\
 0 & 0 & 0 & -1\\
 0 & 0 & 0& -1\\
 0 &  0 & 0& 1\\
\end{pmatrix}.
\end{eqnarray*}
We now ask what is the behaviour of the complementary channel on the subsystem $B$ paired with the fixed state $\sigma_A$; that is, we compute how $\Phi^\sharp$ acts on operators $\sigma_A\otimes\sigma_B$ for all $\sigma_B$. Again, we must be careful: this pairing, which in this case we can identify with the algebra $I_2\otimes \M_2$, is private for $\Lambda'$, and so we wish to test $I_2\otimes \M_2$ on $(\Lambda^\sharp)'$, where we obtain $(\Lambda^\sharp)'$ by applying the unitary transformation $U(\cdot)U^\dagger$, with $U=CNOT_{1,2}CNOT_{2,1}(K^\dagger\otimes I_2)$, as before.

We compute the Kraus operators of $(\Lambda^\sharp)'$ to be $\{B_i=UA_iU^\dagger\}$, where 

\begin{eqnarray*}
B_1&=&\frac14\begin{pmatrix}
1-i &0 & 0 & 1-i\\
1+i &0 & 0 & 1+i\\
1-i &0 & 0 & 1-i\\
1+i &0 & 0 & 1+i\\
\end{pmatrix}
\\
B_2&=&\frac14\begin{pmatrix}
 0 & 1-i & 1-i & 0\\
 0 & -1-i& -1-i & 0\\
 0 & -1+i& -1+i & 0\\
 0 & 1+i&1+i & 0\\
\end{pmatrix}
\\
B_3&=&\frac14\begin{pmatrix}
1-i & 0 & 0 & -1+i\\
-1-i & 0 & 0 & 1+i\\
1-i & 0 & 0 & -1+i\\
-1-i & 0 & 0 & 1+i\\
\end{pmatrix}
\\
B_4&=&\frac14\begin{pmatrix}
 0 & -1+i & 1-i & 0\\
 0 & -1-i& 1+i & 0\\
 0 & 1-i& -1+i & 0\\
 0 & 1+i&-1-i & 0\\\end{pmatrix}.
\end{eqnarray*}

Now, for any $ \sigma_B=\begin{pmatrix} a& b\\c & d\end{pmatrix}\in \m L(B)$, we find 
\[
(\Lambda^\sharp)'\left(\frac12 I_A\otimes \sigma_B\right)=\sum_iB_i\left(\frac12 I_A\otimes \sigma_B\right)B_i^\dagger=\frac14 I_4.
\]

Far from being correctable on the algebra $I_2\otimes \M_2$, the complementary channel $(\Lambda^\sharp)'$ (with the proper unitary transformation) is completely depolarizing. All information is lost, so there is no possibility of the channel being correctable in any sense. In fact, note in this case that the Kraus operators of the complementary map~$\Lambda^{\sharp}$ are four orthogonal rank-one projectors in two-qubit Hilbert space, and in particular the map determines a von Neumann measurement.

However, one can rightly ask if the dephasing map $\Lambda$ and its complementary map are both private, where does the quantum information go? Figure~\ref{fig:DephasingExtension} illustrates the isometric extension of the dephasing channel, along with the encoding of the information from the algebra~$I_2 \otimes \M_2$ to a state of the form of equation~(\ref{subsystem_encoding}).

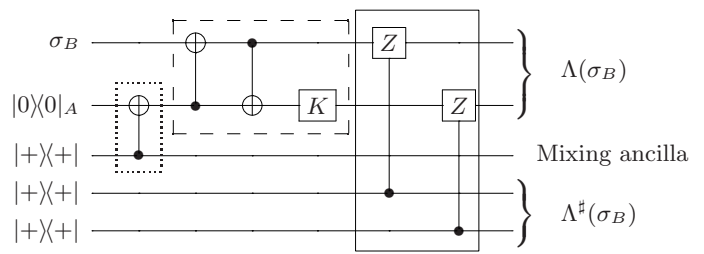
\begin{figure}[htbp]
\begin{align*}
\Qcircuit @C=1.5em @R=1.3em {
\lstick{\sigma_B} & \qw & \targ & \ctrl{1} & \qw & \gate{Z} & \qw &  \qw & \rstick{\raisebox{-3em}{$\Lambda(\sigma_B)$}} \\
\lstick{\kb{0}{0}_A} & \targ & \ctrl{-1} & \targ & \gate{K} & \qw & \gate{Z}  & \qw & \rstick{} \\
\lstick{\kb{+}{+}} & \ctrl{-1} & \qw & \qw & \qw & \qw  &  \qw & \qw & \rstick{\hspace{-1em} \text{Mixing ancilla}} \\
\lstick{\kb{+}{+}} & \qw & \qw & \qw & \qw & \ctrl{-3} & \qw  &  \qw & \rstick{\raisebox{-2.5em}{$\Lambda^\sharp(\sigma_B)$}}\\
\lstick{\kb{+}{+}} & \qw & \qw & \qw & \qw & \qw & \ctrl{-3}  &  \qw & \rstick{}
\gategroup{1}{8}{2}{8}{1.0em}{\}} \gategroup{4}{8}{5}{8}{1.0em}{\}}
\gategroup{2}{2}{3}{2}{1.0em}{.} \gategroup{1}{3}{2}{5}{1.0em}{--} \gategroup{1}{6}{5}{7}{1.4em}{-}
}
\end{align*}
\caption{Isometric extension of the two-qubit dephasing channel $\Lambda = \Lambda_2 \circ \Lambda_1$. The extension of the channel to be a unitary transformation is given in the solid box (---) by introducing ancilla qubits 4 and~5. The dashed box (-- --) contains the encoding of $\rho$ of the form in equation~(\ref{subsystem_encoding}) into the algebra $I_2 \otimes \M_2$. The dotted box ($\cdots$) gives a particular preparation of the mixed state for the subsystem encoding, using a ``mixing ancilla" that is traced out for both the channel~$\Lambda$ and its complementary channel~$\Lambda^{\sharp}$.}
\label{fig:DephasingExtension}
\end{figure}

As Figure~\ref{fig:DephasingExtension} shows, the isometric extension of the channel to a larger Hilbert space, where the state evolution is described by unitary evolution, can be achieved using two extra qubits. Moreover, in order to purify the mixed state used in the subsystem encoding, one could use an additional ``mixing ancilla". Such an ancilla would be traced out both for the dephasing channel~$\Lambda$ and its complementary channel~$\Lambda^{\sharp}$. By definition of the unitary extension of the channel, the channel mapping~$\Lambda(\rho)$ can be obtained by tracing out the final two qubits as well as the mixing ancilla. The complementary channel is obtained by tracing out qubits 1~and~2 as well as the mixing ancilla. As shown above, both of these outputs are private. 

However, what if one had access to the information stored in the mixing ancilla? The role of this state is to twirl the second qubit to obtain a mixed state, however, if one now had access to this state the overall evolution of the channel is no longer on a subsystem encoding, but rather it would be on a subspace encoding that included the mixing ancilla itself. Define $U_{\Lambda}$ as the full unitary evolution described in Figure~\ref{fig:DephasingExtension}, and let $\tilde{\rho} = \rho \otimes \kb{0}{0} \otimes \kb{+}{+} \otimes \kb{0}{0} \otimes \kb{0}{0}$ be the input state into the unitary evolution. The dephasing channel is given by $\text{Tr}_{345} (U_{\Lambda} \tilde{\rho} U_{\Lambda}^\dagger) $, where we trace out qubits 3 (the mixing ancilla), 4, and 5. If we now look at the output on registers $345$, %that is those that are complementary to the output of the dephasing channel, 
the channel~$\text{Tr}_{12} (U_{\Lambda}\tilde{\rho} U_{\Lambda}^\dagger) $ must describe an operator quantum error-correcting code as it is the complement to a channel that is an operator private quantum channel. It is clear that this is not the complementary channel in the sense of the isometric extension, as we are adding on an extra level of operations, namely the mixing of the qubit $\ket{0}$ using a mixing ancilla (and then considering the mixing ancilla as qubit 1 of a larger Hilbert space). However, in this modified notion of the complementary channel we find the quantum information that was lost. We can thus conclude that the mixing ancilla plays an important role in the perseverance of global quantum information, and that the information must be stored in the correlations between this space and one of the two output spaces. 

More generally, given a private subsystem $\m{L}(B)$ for a channel~$\Phi$ (with  fixed mixed state $\sigma_A \in \m{L}(A)$), one can formalize the notion of a correctable complementary channel in a similar fashion. Let $\sigma_A = \sum_{i=1}^N p_i \kb{\psi_i}{\psi_i}_A$ and  define a mixing ancillary Hilbert space $\m{L}(M)$ containing $N$ basis states. The mixing ancillary space $\m{L}(M)$, as in the example, is used to apply a controlled unitary operation $U_i$ based on the state $\ket{\varphi}_A$ in $\m L(A)$,  and the unitaries $U_i$ are chosen such that $U_i (\ket{i}_M\ket{\varphi}_A)  = \ket{i}_M\ket{\psi_i}_A $. That is, 
\begin{align*}
\sum_ip_iU_i\ket{i}_M\kb{\varphi}{\varphi}_A\bra{i}_MU_i^\dagger
& = \sum_ip_i\ket{i}_M\kb{\psi_i}{\psi_i}_A\bra{i}_M\\
&= \sum_ip_i\kb{i}{i}_M\otimes \kb{\psi_i}{\psi_i}_A.
\end{align*} 
Since $\sigma_A=\tr_M\big(\sum_ip_i\kb{i}{i}_M\otimes \kb{\psi_i}{\psi_i}_A\big)$, we find 
\begin{align*}
\sigma_A&=\tr_M\big(\sum_ip_iU_i\ket{i}_M\ket{\varphi}\bra{\varphi}_A\bra{i}_MU_i^\dagger\big)\\
&= \tr_M \big( U_{MA} \kb{\Theta}{\Theta}_M \otimes \kb{\varphi}{\varphi}_A U_{MA}^{\dagger} \big), 
\end{align*}
where $\ket{\Theta}_M = \sum_i \sqrt{p_i} \ket{i}_M$ is a chosen pure state for the mixing ancilla such that the $U_{MA}$ performs the appropriate unitary transformation $U_{MA} = \sum_i \kb{i}{i}_M \otimes \kb{\psi_i}{\varphi}_A$.
%The action of all the controlled unitaries can be summarized by the action of the unitary~$U_{MA}$ arising from Stinespring's dilation theorem on state $\sigma_A$, where we trace out the ancillary mixing space: %, and the summation collapses by defining the state~$\ket{\gamma}_M = \sum_i \sqrt{p_i} \ket{i}_M$ on the mixing ancillary space. The state $\sigma_A$ can then be expressed via Stinespring's dilation by tracing out the ancillary mixing space:
%\begin{align*}
%\sigma_A & = \sum_{i = 1}^N p_i \kb{\psi_i}{\psi_i}_A = \text{Tr}_M \Big( \sum_{i = 1}^N  \kb{i}{i}_M\otimes p_i \kb{\psi_i}{\psi_i}_A \Big) \\
%&= \text{Tr}_M \Big( \sum_{i = 1}^N  \kb{i}{i}_M\otimes  p_i U_i \kb{\phi}{\phi}_A U_i^{\dagger}  \Big)\\%.
%\end{align*}
%Since we are taking the partial trace over $\m L(M)$, the off-diagonal terms in $\m L(M)$ do not factor in, and we can write the above as
%\begin{align*}
%\sigma_A 
%&= \text{Tr}_M \Big( \big( \sum_i \sqrt{p_i}\ket{i}_M U_i \ket{\phi}_A \big) \big( \sum_j \sqrt{p_j}\bra{\phi}_A U_j^{\dagger}\bra{j}_M    \big) \Big) \\
%&= \text{Tr}_M \Big(U_{MA} (\kb{\gamma}{\gamma}_M\otimes \kb{\phi}{\phi}_A ) U_{MA}^{\dagger} \Big).
%\end{align*}
 The private quantum subsystem channel can then be expressed as follows:
\begin{align}
&\Phi(\sigma_A \otimes \sigma_B) = \sum_j A_j (\sigma_A \otimes \sigma_B) A_j^{\dagger} \nonumber \\
&\text{ }=\text{Tr}_K \Big( U_\Phi (\sigma_A \otimes \sigma_B \otimes \kb{\zeta}{\zeta}_K ) U_\Phi^{\dagger} \Big) \nonumber \\
&\text{ }=\tr_{MK} \Big( U_\Phi U_{MA}( \kb{\Theta}{\Theta}_M \otimes \kb{\varphi}{\varphi}_A \nonumber\\
& \qquad \qquad \qquad \qquad \qquad \otimes \sigma_B \otimes \kb{\zeta}{\zeta}_K ) U_{MA}^{\dagger} U_\Phi^{\dagger} \Big).
\label{eq:PrivateExtension}
\end{align}
The transformation within the parenthesis is a unitary transformation, as $U_\Phi$ is a unitary defined by the isometric extension of the channel $\Phi$ (that is, by Stinespring's dilation theorem), where we have introduced the ancillary system~$K$ to form the isometric extension with~$\ket{\zeta}_K$ being a fixed pure state. The unitary $U_{MA}$ corresponds to the transformation in order to prepare a mixed state $\sigma_A$, after tracing out over the mixing ancillary space~$M$. Since the transformation within the brackets is a unitary transformation, if the output state of the channel contains no information about the input state $\sigma_B$, the quantum information must be completely contained in the  traced out subsystem: the $MK$ subsystem. That is, if one traced out the output space, and we were left with the $MK$ subsystem, such an output would necessarily be correctable since all quantum information is contained in that system. That is to say, the generalized conjugate channel
\begin{align}
&\tilde{\Phi}(\sigma_A \otimes \sigma_B) =\text{Tr}_{A\otimes B} \Big( U_\Phi (\sigma_A \otimes \sigma_B \otimes \kb{\zeta}{\zeta}_K ) U_\Phi^{\dagger} \Big) \nonumber \\
& \qquad =\text{Tr}_{AB} \Big( U_\Phi U_{MA}( \kb{\Theta}{\Theta}_M \otimes \kb{\varphi}{\varphi}_A \nonumber \\
& \qquad \qquad \qquad \qquad \qquad \otimes \sigma_B \otimes \kb{\zeta}{\zeta} ) U_{MA}^{\dagger} U_\Phi^{\dagger} \Big)
\label{eq:GeneralizedConjugate}
\end{align}
has the feature that $B$ is error-correctable for it. The generalized form of a private quantum subsystem can thus be summarized as a unitary transformation on an extended Hilbert space by the circuit in Figure~\ref{fig:GeneralExtension}.
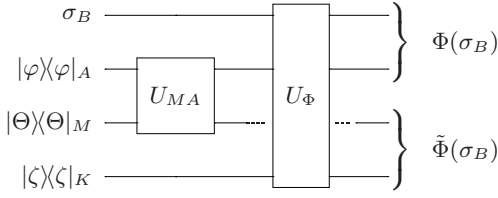
\begin{figure}[htbp]
\begin{align*}
%\centering
%\Qcircuit @C=1.3em @R=1.3em {
%\lstick{\kb{\Theta}{\Theta}_M} & \multigate{1}{U_{MA}} & \qw & \qw \\
%\lstick{\kb{\varphi}{\varphi}_A} & \ghost{U_{MA}} & \multigate{2}{U_{\Phi}}  & \qw \\
%\lstick{\sigma_B} & \qw & \ghost{U_{\Phi}} & \qw \\
%\lstick{\kb{\zeta}{\zeta}_K} & \qw & \ghost{U_{\Phi}}  &  \qw
%}
\Qcircuit @C=1.3em @R=1.3em {
\lstick{\sigma_B} & \qw & \multigate{3}{U_{\Phi}} &  \qw & \rstick{\raisebox{-3em}{$\Phi(\sigma_B)$}} \\
\lstick{\kb{\varphi}{\varphi}_A} & \multigate{1}{U_{MA}} & \ghost{U_{\Phi}}  & \qw & \rstick{} \\
\lstick{\kb{\Theta}{\Theta}_M} & \ghost{U_{MA}} & \push{\hdashrule{1em}{0.4pt}{1pt}\hspace{2.5em}\hdashrule{1em}{0.4pt}{1pt}}  \qw & \qw & \rstick{\raisebox{-3em}{$\tilde{\Phi}(\sigma_B)$}} \\
\lstick{\kb{\zeta}{\zeta}_K} & \qw & \ghost{U_{\Phi}}  &  \qw & \rstick{}
\gategroup{1}{4}{2}{4}{1.0em}{\}} \gategroup{3}{4}{4}{4}{1.0em}{\}}
}
\end{align*}
\caption{Generalized form of extending a private quantum subsystem channel to a unitary transformation via Stinespring's dilation theorem. The subsystem $\sigma_A \otimes \sigma_B$ that encodes the arbitrary state of quantum information~$\sigma_B$ is prepared by entangling an ancillary pure state mixing ancilla~$\kb{\Theta}{\Theta}_M$ with a chosen pure state~$\kb{\varphi}{\varphi}_A$ via the unitary~$U_{MA}$ and tracing out over the mixing ancilla space~$M$. This operation corresponds to the dotted box in Figure~\ref{fig:DephasingExtension}. The action of the private quantum channel~$\Phi$ can also be extended to a unitary transformation over a larger Hilbert space, as described by the action of~$U_\Phi$ on systems $ABK$ by introducing the ancillary state~$\kb{\zeta}{\zeta}_K$, as described in Eq.~\ref{eq:PrivateExtension}. The unitary transformation~$U_\Phi$ corresponds to the dashed and solid boxes in Figure~\ref{fig:DephasingExtension}. The complementary channel is defined on the output space of the extension of the channel~$\Phi$, and therefore corresponds to the final state on system~$K$, yet in general will not be quantum error correctable for an arbitrary subsystem channel. However, a generalized conjugate channel~$\tilde{\Phi}$ can be constructed on the Hilbert space~$MK$, as described in Eq.~\ref{eq:GeneralizedConjugate}, which will necessarily be a quantum error correctable channel since the overall extension is a subspace channel.}
\label{fig:GeneralExtension}
\end{figure}

If we consider the action of the private subsystem channel~$\Phi$ via the isometric extension~$U_{\Phi}$, then the Kraus operators of the original channel can be expressed as follows (without the extension to the mixing ancilla space~$M$):
\begin{align*}
\Phi(\sigma_B) &= \tr_K \big( U_\Phi (\sigma_A \otimes \sigma_B \otimes \kb{\zeta}{\zeta}_K) U_\Phi^{\dagger} \big)\\
&= \sum_i \bra{i}_K  U_\Phi (\sigma_A \otimes \sigma_B \otimes \kb{\zeta}{\zeta}_K) U_\Phi^{\dagger} \ket{i}_K \\
&= \sum_i P_i (\sigma_A \otimes \sigma_B \otimes \kb{\zeta}{\zeta}_K) P_i^{\dagger},
\end{align*}
where the Kraus operators of the channel~$\Phi$ are given by~$\{P_i = \bra{i}_K U_\Phi \}_i$. In a similar manner, the Kraus operators of the complementary channel~$\Phi^{\sharp}$ are given as follows:
\begin{align*}
\Phi^\sharp(\sigma_B) &= \tr_{AB} \big( U_\Phi (\sigma_A \otimes \sigma_B \otimes \kb{\zeta}{\zeta}_K) U_\Phi^{\dagger} \big)\\
&= \sum_j Q_j (\sigma_A \otimes \sigma_B \otimes \kb{\zeta}{\zeta}_K) Q_j^{\dagger},
\end{align*}
where the Kraus operators are given by~$\{Q_j = \bra{j}_{AB} U_\Phi \}_j$. Finally, in order to extend the input space to be a subspace, rather than a subsystem, the ancillary mixing state is introduced. Defining the generalized complementary channel~$\tilde{\Phi}$ as above, the Kraus operators of this channel can be defined on the extended Hilbert space as follows:
\begin{align*}
\tilde{\Phi}(\sigma_B) &= \tr_{AB} \big( U_\Phi U_{MA} (\kb{\Theta}{\Theta}_M \otimes \kb{\varphi}{\varphi}_A \\
&\qquad \qquad \qquad \otimes \sigma_B \otimes \kb{\zeta}{\zeta}_K) U_{MA}^{\dagger} U_\Phi^{\dagger} \big)\\
&= \sum_k \bra{k}_{AB} U_\Phi U_{MA} (\kb{\Theta}{\Theta}_M \otimes \kb{\varphi}{\varphi}_A \\
&\qquad \qquad \qquad \otimes \sigma_B \otimes \kb{\zeta}{\zeta}_K) U_{MA}^{\dagger} U_\Phi^{\dagger} \ket{k}_{AB} \\
&= \sum_k (I_M \otimes Q_k) U_{MA} (\kb{\Theta}{\Theta}_M \otimes \kb{\varphi}{\varphi}_A \\
&\qquad \qquad \qquad \otimes \sigma_B \otimes \kb{\zeta}{\zeta}_K) U_{MA}^{\dagger} (I_M \otimes Q_k)^{\dagger}\\
&= \sum_k R_k (\kb{\Theta}{\Theta}_M \otimes \kb{\varphi}{\varphi}_A \\
&\qquad \qquad \qquad \otimes \sigma_B \otimes \kb{\zeta}{\zeta}_K) R_k^{\dagger},
\end{align*}
where the Kraus operators of the generalized complementary channel~$\tilde{\Phi}$ mapping to the~$MK$~Hilbert space are related to the Kraus operators of the complementary channel~$\Phi^\sharp$ via the relationship~$\{ R_k = (I_M \otimes Q_k) U_{MA} \}_k$. As outlined above, this channel must be quantum error correcting, and as such must satisfy the Knill-Laflamme conditions~\cite{KL97}:
\begin{align*}
\bra{i}_B R_p^{\dagger} R_q \ket{j}_B &= \bra{i}_B U_{MA}^{\dagger} (\id_K \otimes Q_p^{\dagger}) (\id_K \otimes Q_q) U_{MA} \ket{j}_B \\
&=  \bra{i_{U_{MA}}}_{MA} (\id_K \otimes Q_p^{\dagger}) (\id_K \otimes Q_q)  \ket{j_{U_{MA}}}_{MA} \\
&= \delta_{ij} c_{pq}.
\end{align*}
For any generalized private subsystem channel there must be the existence of a higher dimensional Hilbert space such that the above Knill-Laflamme conditions for quantum error correcting hold for a set of Kraus operators related to the Kraus operators of the complementary channel of the original private subsystem channel.

\subsection{Generalized channels on subspace and subsystem encodings}
A common theme throughout this work has been that encoding into a subsystem, rather than a subspace, generates an increased freedom in the types of channels that can be used to privatize quantum information. In this section we explore this notion further, explicitly showing that the set of unitaries that can be used to privatize quantum information in a subsystem code are inherently richer than those for subspace codes. We shall focus on the case of encoding a single qubit of information into either a two-qubit subspace or a two-qubit subsystem.

Consider an arbitrary encoding of a single qubit into a two-qubit subspace:
\begin{align*}
(\alpha \ket{0} + \beta \ket{1}) \ket{0} \rightarrow \alpha \ket{0_L}_{12} + \beta{\ket{1_L}_{12}},
\end{align*}
where $\ket{0_L}$ and~$\ket{1_L}$ represent the logically encoded states in a higher dimensional Hilbert space. An arbitrary CPTP map can be described as a transformation of the encoded basis states to a larger dimensional Hilbert space, after which a trace is taken in the environment. Any arbitrary transformation can be described as follows:
\begin{align*}
\ket{0_L}_{12} \rightarrow \sum_{ij} \ket{ij}_{12} \ket{E_{ij}^0}_E \\
\ket{1_L}_{12} \rightarrow \sum_{ij} \ket{ij}_{12} \ket{E_{ij}^1}_E,
\end{align*}
where the states $\ket{E_{ij}^m}$ are arbitrary environment states, that are not necessarily normalized or orthogonal. The environment states are the states on the ancillary qubits when expressing the final state in the computational basis of the first two qubits. Tracing out over the environment states, the resulting entries of the two-qubit mixed states have a particular form:
\begin{align}
&\kb{ij}{kl}_{12} \tr_E \big(\alpha \ket{E_{ij}^0} + \beta \ket{E_{ij}^1})(\alpha^* \bra{E_{kl}^0} + \beta^* \bra{E_{kl}^1}\big)  \nonumber \\
=& \kb{ij}{kl}_{12} \tr_E \big( |\alpha|^2 \kb{E_{ij}^0}{E_{kl}^0} + \alpha \beta^* \kb{E_{ij}^0}{E_{kl}^1} \nonumber \\
&\qquad \qquad \qquad + \alpha^* \beta \kb{E_{ij}^1}{E_{kl}^0} + |\beta|^2 \kb{E_{ij}^1}{E_{kl}^1} \big) , 
\label{eq:PureConditions}
\end{align}
this imposes a set of conditions on the environmental states in order for the output on the first two qubits to be private, namely the terms after tracing out can yield no information about the input state as described by~$\alpha$ and~$\beta$.

Consider now the same isometric extension mapping along with the inclusion of a third qubit that will serve as a mixing ancilla. The encoding operation is now generalized to a three qubit encoding, which upon tracing out the mixing ancilla will return the subsystem encoding on qubits 1~and~2. The generalized mapping is modified to include the third qubit.
\begin{align*}
\ket{0_L}_{123} \rightarrow \sum_{ij} \ket{ijk}_{123} \ket{E_{ijk}^0}_E \\
\ket{1_L}_{123} \rightarrow \sum_{ij} \ket{ijk}_{123} \ket{E_{ijk}^1}_E.
\end{align*}
The generalized form of the mixed state entries on the first two qubits thus have the form
\begin{align}
&\kb{ij}{kl}_{12} \tr_{3E} \Big( \big( \sum_p \alpha \ket{p}_3\ket{E_{ijp}^0} + \beta \ket{p}_3\ket{E_{ijp}^1}\big) \nonumber \\
& \qquad \qquad \qquad \big( \sum_q \alpha^* \bra{q}_3 \bra{E_{klq}^0} + \beta^* \bra{q}_3 \bra{E_{klq}^1} \big) \Big) \nonumber \\
=& \kb{ij}{kl}_{12} \tr_{E} \Big(  |\alpha|^2 ( \kb{E_{ij0}^0}{E_{kl0}^0} + \kb{E_{ij1}^0}{E_{kl1}^0}) \nonumber \\
&\qquad \qquad \qquad + \alpha \beta^* (\kb{E_{ij0}^0}{E_{kl0}^1} + \kb{E_{ij1}^0}{E_{kl1}^1} ) \nonumber \\
&\qquad \qquad \qquad + \alpha^* \beta (\kb{E_{ij0}^1}{E_{kl0}^0} + \kb{E_{ij1}^1}{E_{kl1}^0}) \nonumber \\
&\qquad \qquad \qquad + |\beta|^2 (\kb{E_{ij0}^1}{E_{kl0}^1} + \kb{E_{ij1} ^1}{E_{kl1}^1})  \Big), 
\label{eq:MixedConditions}
\end{align}
therefore, by comparing equations~\ref{eq:PureConditions}~and~\ref{eq:MixedConditions}, we find that in the case where a mixing ancilla has been introduced the set of conditions upon privatizing the output on the first two qubits is looser in terms of the environment states. Namely there is a freedom in choosing the environment states such that certain terms can cancel out to yield no information; this freedom does not exist in the case of a pure state encoding. We explore these set of conditions in more detail in Appendix~\ref{app:Conditions}.

%%%%%%%%%%
%%%%%%%%%%
\section{Quantum Error Correction Revisited}
\label{sec:GenOQECC}

In this Section, we revisit the notion of an operator quantum error correctable subsystem and its parallels to private quantum subsystems. We begin with the definition of an operator quantum error-correcting code \cite{KLP05,KLPL06}. 

\begin{definition}\label{defn:OQECC}
Let $S=(A\otimes B)\oplus(A\otimes B)^\perp$ and let $\m E$ be a channel acting on $\m L(S)$. Then 
 $B$ is an \emph{operator quantum error correcting code (OQECC)} for $\mathcal E$ if there exists a quantum channel $\mathcal R$ such that for all $\sigma_A$, for all $\sigma_B$, there exists  some fixed state $\tau_A=\tau_A(\sigma_A)$ (dependent on $\sigma_A$) such that
 \[
 \mathcal R \circ \mathcal E (\sigma_A \otimes \sigma_B) = \tau_A \otimes \sigma_B.
 \]

\end{definition}

The discussion of private subsystem versus operator private subsystem in this work and \cite{JKLP12} motivates the following observation: The notion of an operator quantum error-correcting subsystem can be expanded to mimic the general definition of a private quantum subsystem. We proposed the following definition in \cite{JKLP12}, which can be seen as the QEC analogue of equation~\eqref{eq:private_definition}:

\begin{definition}\label{defn:GenOQECC}
Let $S=(A\otimes B)\oplus(A\otimes B)^\perp$ and let $\m E$ be a channel acting on $\m L(S)$. Then 
 $B$ is a \emph{generalized operator quantum error correcting code (GenOQECC)} for $\mathcal E$ if there exists a quantum channel $\mathcal R$ for which there exists  a fixed state  $\sigma_A$ and a state $\tau_A=\tau_A(\sigma_A)$ (dependent on $\sigma_A$) such that for all $\sigma_B$, we have
 \[
 \mathcal R \circ \mathcal E (\sigma_A \otimes \sigma_B) = \tau_A \otimes \sigma_B.
 \]
Clearly no generality is lost in this definition by setting $\tau_A = \sigma_A$. 
\end{definition}

\begin{ex}
{\rm 
Consider the following example of a generalized operator quantum error correcting code. Let 
\[
\sigma_A = (1-4p)\kb{0000}{0000} + p \sum_{\text{wt}(x)=1}\kb{x}{x}
\] 
be the fixed ancilla state, a mixed 4-qubit state, where the states $\ket{x}$ are the set of (four) computational basis states with Hamming weight 1. The weighting~$p$ can be thought of as a probability of failure of preparing a desired ground state~$\ket{0}$ for the purpose of error correction, where we have omitted higher order~$p$ terms. Let $\sigma_B$ be any single qubit state. The encoding of the subsystem code is a controlled operation from qubit~$B$ which targets all qubits of the state~$\sigma_A$ with a controlled-$X$, we shall call such an encoding operation~$U_{AB}$.
The error map will be the probabilistic application of an $X$ error on any of the 5 encoded qubits given by the set of Kraus operators~$\{ \sqrt{\epsilon_i}X_i \}_{i=0}^5$, where~$\epsilon_i$ is the probability of the error~$X_i$ occurring ($X_0$ denoting the identity operation). The application of such an error map will produce the following mapping on the encoded state for an arbitrary~$B$ state $\ket{\psi}_B = \alpha \ket{0} + \beta \ket{1}$,
\begin{align*}
&\m E \big( U_{AB} (\sigma_A \otimes \kb{\psi}{\psi}_B) U_{AB}^{\dagger} \big) \\
&= \m E \Big(  |\alpha|^2 \big( (1-4p)\kb{00000}{00000} + p \sum_{\text{wt}(x)=1} \kb{x0}{x0} \big) \\
& \qquad +  \alpha \beta^* \big( (1-4p)\kb{00000}{11111} + p \sum_{\text{wt}(x)=1} \kb{x0}{\overline{x}1} \big) \\
& \qquad +  \beta \alpha^* \big( (1-4p)\kb{11111}{00000} + p \sum_{\text{wt}(x)=1} \kb{\overline{x}1}{x0} \big) \\
& \qquad +  |\beta|^2 \big( (1-4p)\kb{11111}{11111} + p \sum_{\text{wt}(x)=1} \kb{\overline{x}1}{\overline{x}1} \big) \Big) \\
&=  \sum_{i=0}^5 \Big(  |\alpha|^2 \big( (1-4p)X_i\kb{0}{0}^{\otimes 5}X_i \\
& \qquad + p \sum_{\text{wt}(x)=1} X_i\kb{x0}{x0}X_i \big) \\
& \text{ } +  \alpha \beta^* \big( (1-4p)X_i\kb{0}{1}^{\otimes 5} X_i + p \sum_{\text{wt}(x)=1} X_i\kb{0}{\overline{x}1}X_i \big) \\
& \text{ } +  \beta \alpha^* \big( (1-4p)X_i\kb{1}{0}^{\otimes 5} X_i + p \sum_{\text{wt}(x)=1} X_i\kb{\overline{x}1}{x0}X_i \big) \\
& \text{ } +  |\beta|^2 \big( (1-4p)X_i\kb{1}{1}^{\otimes 5} X_i + p \sum_{\text{wt}(x)=1} X_i\kb{\overline{x}1}{\overline{x}1}X_i \big) \Big), 
\end{align*}
where we have defined $\ket{\overline{x}} = X^{\otimes 5} \ket{x}$. One can notice that the error map will flip at most one bit. This is important as the encoded~$\ket{0}$ terms have weight~0~or~1 for all terms, while the encoded~$\ket{1}$ have weight~4~or~5. This means that after the application of the error map, the encoded~$\ket{0}$ will have a weight between 0~and~2, while the encoded~$\ket{1}$ will have weight between 3~and~5. The recovery operation will then perform a weight check using measurement in the computational basis, associating all states with weight~$\le 2$ to an encoded~$\ket{0}$ and all states with weight~$\ge 3$ to an encoded~$\ket{1}$ state. As such all~$X_i$ errors are corrected. Since this error correction procedure works for an arbitrary pure state encoding of~$\sigma_B$, it will necessarily work for the full set of states in $\m L(B)$. That is, $B$ is a generalized operator quantum error correcting code for $\m E$. 

%As long as only one of the qubits has a preparation error, the value of $p$ does not matter. Thus, the channel $\m E$ is correctable for all $\sigma_B$ and for all $\sigma_A$ of the given form; however, it is not correctable for any arbitrary $\sigma_A\in \m L(A)$ (that is any 4-qubit state in the ancillary space). Thus, we do not have operator quantum error correction (or standard subspace QEC) in the strict sense of the definition. 

It is worth noting that the error correction procedure does not work if we chose the ancillary mixed state to be outside the set of states of weight 0~or~1. Consider a particular example of a 4-qubit state of weight~2, given by $\sigma_A= \kb{1100}{1100}_A$. We shall show that encoding using such an ancillary state will not correct the error map for a particular choice of $\kb{\psi}{\psi}_B = \kb{0}{0}_B$. The action of the error map is as follows:
\begin{align*}
& \m R \circ \m E \big( U_{AB} \kb{11000}{11000} U_{AB}^{\dagger} \big) \\
&=\m R \circ \m E \big(\kb{11000}{11000} \big) \\
&=\m R \big( \sum_{i=0}^5 \epsilon_i X_i \kb{11000}{11000} X_i \big) \\
&= (\epsilon_0 + \epsilon_1 + \epsilon_2) \tau_{A,0} \otimes \kb{0}{0}_B + (\epsilon_3 + \epsilon _4 + \epsilon_5) \tau_{A,1} \otimes \kb{1}{1}.
\end{align*}
The recovery operator maps all states that correspond to either no error or errors on the first 2 qubits to the correct state~$\kb{0}{0}$, this is since the action of the error map returns a state of weight~1~or~2. However, for an error that occurs on qubits 3 through~5, the state state before the action of the recovery operator is now of weight~3, which will then be mapped to the state~$\kb{1}{1}$ by the definition of the action of recovery operator. Thus, as long as there is a non-zero probability of an error on the last~3 qubits, the action of the error map will result in the recovery of an incorrect state, as the state does not have the form~$\tau_A \otimes \sigma_B$. Similarly, for any choice of ancillary state $\sigma_A$ of weight greater or equal to 2 there will exist a state $\kb{\psi}{\psi}$ that will result in faulty error correction.

Therefore, we know that the GenOQECC corrects for the error map $\{ \sqrt{\epsilon_i}X_i \}_{i=0}^5$ for the given fixed state, and will not be error correcting for ancillary states with weight greater or equal to 2. However, it is worth noting that as long as only one of the qubits has a preparation error (therefore weight~1), the value of $p$ does not matter. Thus, the channel $\m E$ is correctable for all $\sigma_B$ and for all $\sigma_A$ of weight~0 or~1; but, it is not correctable for any arbitrary $\sigma_A\in \m L(A)$ (that is any 4-qubit state in the ancillary space). Hence, if we consider the full 4-qubit ancillary Hilbert space $\m L(A)$, we will not have an operator quantum error correcting code (OQEC) on such a space. 

}
\end{ex}

%It is clear that the existence of an operator quantum error correcting code  implies the existence of a generalized operator quantum error correcting code for any fixed $\sigma_A$. The following theorem shows that if $\m E$ has a generalized operator quantum error correcting code $B$ of a certain dimension as in Definition \ref{defn:GenOQECC}, then there is an OQECC for $\m E$ of the same dimension; namely, a subspace $A'$ of $A$ can be found such that the pairing $A' \otimes B$ defines an OQECC for $\m E$ on the subsystem $B$. 

This example shows that a GenOQEC code may not be an OQEC code for a given error map. Nevertheless, the following result shows that whenever a GenOQECC exists, we can still find an OQECC for the map of the same dimension. To find such a code we must consider the ancilla $A$ more carefully.  	
	
\begin{theorem}\label{thm:GOQECtoOQEC}
Given a decomposition $S = (A\otimes B) \oplus (A\otimes B)^\perp$ and
channel $\m E$ on $\m L(S)$, suppose 
there exists $\sigma_A$ and channel $\m R$ on $\m L(S)$ such that for all $\sigma_B$, 
\[
\m R\circ\m E(\sigma_A\otimes\sigma_B) = \sigma_A \otimes \sigma_B.\]
Then there exists $\ket{\alpha}\in A$ and channel $\m R_\alpha$ on $\m L(S)$ such that for all $\sigma_B$,  
%\[
%\forall\sigma_{A'}\, \forall\sigma_B,\ \exists
%\tau_{A'}\ :\ \m R\circ\m E(\sigma_{A'}\otimes\sigma_B) = \tau_{A'}\otimes\sigma_B.\]
\[
\m R_\alpha\circ\m E(\kb{\alpha}{\alpha}\otimes\sigma_B) = \ket{\alpha}\bra{\alpha}\otimes\sigma_B;
\]
in other words, the subspace $\alpha \otimes B$ is an error-correcting code for $\m E$. 
\end{theorem}

%In words, if a subsystem $B$ is correctable for a channel $\m E$ for a particular $\sigma_A\in \m L(A)$, then $B$ is correctable for a channel $\m E$ for any $\sigma_{A'}\in \m L({A'})\subseteq \m L(A)$. We shall see in the proof of this theorem that ${A'}$ is in fact the span of the eigenspaces of  $\sigma_A$.

\begin{proof}  First let $\ket{\psi}\in B$ and put $P =
\kb{\psi}{\psi}$. Let $\{\ket{\alpha_k}\}$ be the normalized eigenvectors of $\sigma_A$ so that  $\sigma_A=\sum_{k=1}^m
p_k\kb{\alpha_k}{\alpha_k}$ where $0< p_k\leq 1$. By assumption and using the
positivity of $\m R \circ\m E$ we have for all $k$,
\begin{eqnarray*}
0 &\leq& \m R \circ\m E(p_k\kb{\alpha_k}{\alpha_k}\otimes P) \\
&=&p_k\m R \circ\m E( \kb{\alpha_k}{\alpha_k} \otimes P) \\
&\leq& \m R \circ\m E (\sigma_A
\otimes P) \\ 
%\Rightarrow 
%0 &\leq& \m E( \kb{\alpha_k}{\alpha_k} \otimes P)  \leq \m E (\sigma_A
%\otimes P) \\ 
&=& \sigma_A \otimes P \\   &=& (I_A\otimes P)
(\sigma_A \otimes P)(I_A \otimes P).
\end{eqnarray*}
It follows that there are positive operators $\sigma_{\psi,k}\in
\m L(A)$ such that $\m R \circ\m E(p_k\kb{\alpha_k}{\alpha_k}\otimes P) =
\sigma_{\psi,k} \otimes P$ for all $k$. We can trace-normalize to write $\m R \circ\m E(\kb{\alpha_k}{\alpha_k}\otimes P) =\sigma_{\psi,k} \otimes P$ for all $k$, where $\sigma_{\psi,k}$ are now density operators. 

In fact, the operators $\sigma_{\psi,k}$ do not depend on
$\ket{\psi}$. To verify this claim,  for brevity we shall assume 
$\dim B =2$. The case of general $B$ easily follows. So
let $\ket{\psi_i}$, $i=1,2$, be an orthonormal basis for $B$.
Let $P_i = \kb{\psi_i}{\psi_i}$, $i=1,2$, and put $P_{\pm} =
\kb{\pm}{\pm}$ where $\ket{\pm} = \frac{1}{\sqrt{2}}(\ket{\psi_1}
\pm \ket{\psi_2})$. Fix $\alpha=\alpha_k$. By the above argument,
there are operators $\sigma_{\pm,\alpha}$ and $\sigma_{i,\alpha}$
on $A$ such that
\begin{eqnarray*}
& \m R \circ \m E(\kb{\alpha}{\alpha}\otimes P_{\pm}) = \sigma_{\pm,\alpha}
\otimes P_{\pm} \\
\textnormal{and} \quad & \m R \circ\m E(\kb{\alpha}{\alpha}\otimes P_{i}) =
\sigma_{i,\alpha} \otimes P_{i}.
\end{eqnarray*}
In particular, as $I_B = P_+ + P_- = P_1 + P_2$, we have
\begin{eqnarray*}
 \m E(\kb{\alpha}{\alpha} \otimes I_B) &=&
\sigma_{1,\alpha} \otimes P_{1} + \sigma_{2,\alpha} \otimes P_{2}
\\ &=& \sigma_{+,\alpha} \otimes P_{+} + \sigma_{-,\alpha} \otimes
P_{-}.
\end{eqnarray*}
If we compress this equation by the projection $I_A \otimes
P_1$, we obtain
\begin{align*}
&(I_A\otimes P_1) \m E(\kb{\alpha}{\alpha} \otimes I_B) (I_A
\otimes P_1)= \sigma_{1,\alpha} \otimes P_{1}
\\ &\quad=\frac{1}{2} (\sigma_{+,\alpha} + \sigma_{-,\alpha}) \otimes
P_{1}.
\end{align*}
Thus, $\sigma_{1,\alpha} =  \frac{1}{2} (\sigma_{+,\alpha} +
\sigma_{-,\alpha})$ and since the same identity holds for
$\sigma_{2,\alpha}$ when we compress by $I_A\otimes P_2$, we
obtain $\sigma_{1,\alpha} = \sigma_{2,\alpha}$. There is nothing particularly special about our use of 
$\ket{\pm}$ here, and in fact this argument may be applied to show the same operator is obtained for any pure state on $A$. 

The proof is now completed by a simple linearity argument. Indeed, write $\sigma_k := \sigma_{\psi,k}$, so we have $\m R \circ\m E(\ket{\alpha_k}\bra{\alpha_k}\otimes P_\psi) = \sigma_k \otimes P_\psi$, and by linearity $P_\psi$ can be replaced by an arbitrary $\sigma_B$. We may then choose a channel $\m R_k$ such that $\m R_k(\sigma_k \otimes \sigma_B) = \ket{\alpha_k}\bra{\alpha_k}\otimes \sigma_B$ for all $\sigma_B$. It follows that $\ket{\alpha_k}\otimes B$ is correctable for $\m E$, with a recovery operation given by $\m R_k \circ \m R$. 

%As % $\ket{\alpha}$ and 
%$\ket{\psi_i}$, $i=1,2$, were chosen arbitrarily, we conclude that
%\[
%\forall\ket{\alpha_k}\in {A'},\, \forall\sigma_B,\ \exists
%\tau_{A', k}\ \textnormal{ and } \m R_k:\ \m R_k\circ\m E(\kb{\alpha_k}{\alpha_k}\otimes\sigma_B) = %\tau_{A', k}\otimes\sigma_B.\]

%
%Let $\m H_{A'}=\operatorname{span}\{\ket{\alpha_k}\}$. 
%Now, let $\sigma_A=\sum_{k=1}^m
%p_k\kb{\alpha_k}{\alpha_k}$ where $0< p_k\leq 1$ so that $\sigma_{A'}\in \mathfrak{B}(\m H_A)$. Then we have for all $\sigma_B$,
%\begin{eqnarray*}
%\m E(\sigma_A\otimes \sigma_B)&=&\sum_{k=1}^mp_k\m E(\kb{\alpha_k}{\alpha_k}\otimes\sigma_B),
%\end{eqnarray*}
%where each term in this summation is correctable with recovery operator $\m R_k$. Thus, although we do not have a universal recovery operator to apply to $\m E(\sigma_A\otimes \sigma_B)$, we can still correct each individual $\m E(\kb{\alpha_k}{\alpha_k}\otimes\sigma_B)$ to obtain $\sum_{k=1}^mp_k\tau_{A'}_k\otimes\sigma_B$, which is of the form $\tau_{A'}\otimes \sigma_B$. 
\end{proof}

We note that the above argument can be adjusted to show that in fact any eigenspace $A'$ for $\sigma_A$ determines an OQEC code (which will be a subsystem when $\dim A' > 1$) for the error map of the same size, via the pairing $A'$ and $B$.

\section{Conclusion}
\label{sec:Conclusion}

Private quantum subsystems are subsystem encodings of quantum information that are privatized under the action of a given channel. In this work we have expanded upon the results presented in~\cite{JKLP12} on private quantum subsystems, by providing proofs of the results therein and expanding the analysis of the main example  presented in~\cite{JKLP12}, namely the multi-qubit dephasing channel. We have added to our analysis showing that the multi-qubit dephasing channel has a private subsystem without exhibiting a private subspace. We explicitly showed that this private subsystem is not operator private, which is the first such example we are aware of. Additionally, we have revisited the set of testable algebraic conditions for private quantum subsystems, expanding the discussion of examples and providing further results for particular forms of the channels and output states. 

One of the surprising structural aspects of the most general private quantum subsystem channels (in contrast to operator private subsystems) is that their complementary channel, obtained through the Stinespring dilation, is no longer necessarily correctable on the subsystem, and can in fact be private. In this work we have given an analysis and discussion of where the quantum information is leaking to in such a setting by dilating to an even higher dimensional Hilbert space than is required by the usual notion of complementary channels.

Finally, the difference between general private subsystems and operator private subsystems presented in ~\cite{JKLP12} suggested there could be an analogous notion of generalized operator quantum error correction. In this work we provided an explicit definition for these codes and showed that the existence of such a code implies the existence of a standard (subspace) QEC code of the same size, determined by the fixed ancillary state used in the generalized code. Thus, the generalized notion does not lead to larger codes than what can be found in standard QEC. However, the same is true when one compares operator QEC to standard QEC; indeed, this is even obvious from the definitions of the two code types in that case. What generated significant interest in OQEC codes beyond the theoretical appeal of the mathematical framework, was that it turned out such codes can have extra features that make them quite valuable for fault tolerant quantum computing. It would be interesting to know whether generalized QEC codes have similar advantages. A next step in the analysis would be to broaden the set of generalized code examples which are neither subspaces nor operator subsystems. We leave these and other investigations on the topic for elsewhere.

\section{Acknowledgments}

T. J.-O. was supported by the Ontario Ministry of Training, Colleges and Universities and the Fonds de recherche du Qu\'ebec -- Nature et technologies. D.W.K. was supported by NSERC and a Guelph University Research Chair. R.L. was supported by NSERC, CIFAR, and Industry Canada. S.P. was partially supported by an NSERC Graduate Scholarship.

\begin{appendix}

\section{Conditions on a generalized private quantum channel}
\label{app:Conditions}
The generalized private quantum channel on a two-qubit subspace encoding of a single qubit of information was shown in Equation~\ref{eq:PureConditions} to have the following density matrix output:
\begin{align*}
&\kb{ij}{kl}_{12} \tr_E \big( |\alpha|^2 \kb{E_{ij}^0}{E_{kl}^0} + \alpha \beta^* \kb{E_{ij}^0}{E_{kl}^1} \nonumber \\
&\qquad \qquad \qquad + \alpha^* \beta \kb{E_{ij}^1}{E_{kl}^0} + |\beta|^2 \kb{E_{ij}^1}{E_{kl}^1} \big). 
\end{align*}
For such an output to be private, the output state must encode no information about the input state of the channel, therefore must yield no information about the coefficients~$\alpha$ and~$\beta$. The cross terms ($\alpha \beta^*$ and~$\alpha^* \beta$) must always be zero as there is no choice of non-zero overlap between the states $\ket{E_{ij}^0}$ and $\ket{E_{kl}^1}$ that will cancel out all information stored in~$\alpha \beta^*$ and its complex conjugate for arbitrary choices of~$\alpha$ and~$\beta$. This should come as no surprise, as the complementary channel should be quantum error correcting for a subspace code, and an overlap of $\ket{E_{ij}^0}$ and~$\ket{E_{kl}^1}$ would violate the Knill-Laflamme conditions.

For any non-zero $|\alpha|^2$ term $\bk{E_{ij}^0}{E_{kl}^0}$, the corresponding $|\beta|^2$ term $\bk{E_{ij}^1}{E_{kl}^1}$ must be the same as otherwise information about the relative magnitude of $\alpha$~and~$\beta$ will be contained in the output density matrix, yielding a non-private channel. Therefore, the set of conditions for a private quantum channel on a subspace encoding can be summarized by the following conditions on the environment states when considering the isometric extension:
\begin{align}
\bk{E_{ij}^0}{E_{kl}^0} &= \bk{E_{ij}^1}{E_{kl}^1}, \label{eq:SubspaceConditions1} \\
\bk{E_{ij}^0}{E_{kl}^1} &= 0, \label{eq:SubspaceConditions2}
\end{align}
for all choices of $i, j, k, l$.

The output density matrix in the case of a subsystem encoding, with respect to an enlarged set of environment states is given in Equation~\ref{eq:MixedConditions} as follows:
\begin{align*}
& \kb{ij}{kl}_{12} \tr_{E} \Big(  |\alpha|^2 ( \kb{E_{ij0}^0}{E_{kl0}^0} + \kb{E_{ij1}^0}{E_{kl1}^0}) \nonumber \\
&\qquad \qquad \qquad + \alpha \beta^* (\kb{E_{ij0}^0}{E_{kl0}^1} + \kb{E_{ij1}^0}{E_{kl1}^1} ) \nonumber \\
&\qquad \qquad \qquad + \alpha^* \beta (\kb{E_{ij0}^1}{E_{kl0}^0} + \kb{E_{ij1}^1}{E_{kl1}^0}) \nonumber \\
&\qquad \qquad \qquad + |\beta|^2 (\kb{E_{ij0}^1}{E_{kl0}^1} + \kb{E_{ij1} ^1}{E_{kl1}^1})  \Big), 
\end{align*}
The increased freedom in choosing this output state to be private comes from the fact that in considering the cross terms ($\alpha \beta^*$ and~$\alpha^* \beta$), while the trace over the set of states corresponding to these terms must be zero, there can be cancellation between the two corresponding terms. Therefore, unlike the case of a subspace encoding, one could have $\bk{E_{ij0}^0}{E_{kl0}^1} \ne 0$, however its corresponding pair must cancel the term out, that is~$\bk{E_{ij0}^0}{E_{kl0}^1} =  - \bk{E_{ij1}^0}{E_{kl1}^1}$. There is additional freedom in the diagonal terms in order for no information about the magnitude of the amplitudes of $\alpha$ and $\beta$ to be present in the output density matrix. The result conditions for privatization are summarized as follows:

\begin{align}
\bk{E_{ij0}^0}{E_{kl0}^0} + \bk{E_{ij1}^0}{E_{kl1}^0} &= \bk{E_{ij0}^1}{E_{kl0}^1} + \bk{E_{ij1}^1}{E_{kl1}^1}, \label{eq:SubsystemConditions1} \\
\bk{E_{ij0}^0}{E_{kl0}^1} &= - \bk{E_{ij1}^0}{E_{kl1}^1}, \label{eq:SubsystemConditions2}
\end{align}
for all choices of $i, j, k, l$.

Applying the above set of conditions to the case of the two-qubit dephasing channel~$\Lambda$ described throughout this work, we can show that there is insufficient freedom in a two-qubit subspace encoding to privatize a single encoded qubit. That is, no two-qubit subspace encoding will satisfy Equations~\ref{eq:SubspaceConditions1}~and~\ref{eq:SubspaceConditions2} for the environment states produced by the two-qubit dephasing channel~$\Lambda$.

Let the following parameters denote an arbitrary two-qubit subspace encoding:
\begin{align*}
\ket{0_L} = \alpha_{00} \ket{00} +  \alpha_{01} \ket{01}  +  \alpha_{10} \ket{10} +  \alpha_{11} \ket{11} \\
\ket{1_L} = \beta_{00} \ket{00} +  \beta_{01} \ket{01}  +  \beta_{10} \ket{10} +  \beta_{11} \ket{11} .
\end{align*}
By the uniqueness of the Stinespring dilation Theorem up to the preparation of the ancillary states, we assume that the form extension of the channel to unitary transformation on a larger Hilbert space by preparing an additional pair of qubits in the $\ket{+}$~state and performing controlled--$Z$ operations on each corresponding physical qubit in the encoding, as described in solid boxed operation in Figure~\ref{fig:DephasingExtension}. The resulting mapping of the logical states is given as follows:
\begin{align*}
\ket{0_L} = \sum_{ij} \alpha_{ij} \ket{ij} \longrightarrow \sum_{ij} \alpha_{ij} \ket{ij} Z^i \ket{+} Z^j \ket{+} \\
\ket{1_L} = \sum_{ij} \beta_{ij} \ket{ij} \longrightarrow \sum_{ij} \beta_{ij} \ket{ij} Z^i \ket{+} Z^j \ket{+},
\end{align*}
where the operation~$Z^i$ is applied to the state~$\ket{+}$ depending on the value of the state on qubit~1, and similarly for~$Z^j$ and qubit~2. The resulting environment states therefore have the form
\begin{align*}
\ket{E_{ij}^0} = \alpha_{ij} Z^i \ket{+} Z^j \ket{+} ,\\
\ket{E_{ij}^1} = \beta_{ij} Z^i \ket{+} Z^j \ket{+}.\\
\end{align*}
The conditions set by Equation~\ref{eq:SubspaceConditions2} impose restrictions on the values of the coefficients in the subspace encodings. Since $\bk{E_{ij}^0}{E_{ij}^1} = \alpha_{ij}^* \beta_{ij} = 0$, this implies either $\alpha_{ij} = 0$ or $\beta_{ij} = 0$. Without loss of generality, suppose $\alpha_{ij} = 0$, then the corresponding condition set by Equation~\ref{eq:SubspaceConditions1} imply~$\bk{E_{ij}^0}{E_{ij}^0} = |\alpha_{ij}|^2 = 0 = |\beta_{ij}|^2 =  \bk{E_{ij}^1}{E_{ij}^1}$. Thus, for all values~$(ij)$ the associated coefficients~$\alpha_{ij}$ and~$\beta_{ij}$ will be equal to~0, implying that no private subspace encoding exists for the dephasing channel~$\Lambda$ that satisfy the set of conditions outlined by Equations~\ref{eq:SubspaceConditions1}--\ref{eq:SubspaceConditions2}.

We now show that the set of conditions on a two-qubit subsystem encoding, by introducing a mixing ancilla, can be satisfied by the chosen encoding given by the first two boxes in Figure~\ref{fig:DephasingExtension}. 
\begin{align*}
\ket{0_L} = \frac{1}{2} ( \ket{000} + i \ket{010} + \ket{101} + i \ket{111}) \\
\ket{1_L} = \frac{1}{2} ( \ket{100} - i \ket{110} + \ket{001} - i \ket{011}) \\
\end{align*}
The resulting mapping as given by the isometric extension of the channel by introducing two ancillary~$\ket{+}$ states and controlled--$Z$ operations will have the form
\begin{align*}
\ket{0_L} = \sum_{ijk} \gamma_{ijk} \ket{ijk} \longrightarrow \sum_{ijk} \gamma_{ijk} \ket{ijk} Z^i \ket{+} Z^j \ket{+} \\
\ket{1_L} = \sum_{ijk} \eta_{ijk} \ket{ijk} \longrightarrow \sum_{ijk} \eta_{ijk} \ket{ijk} Z^i \ket{+} Z^j \ket{+},
\end{align*}
resulting in the environment states
\begin{align*}
\ket{E_{ijk}^0} = \gamma_{ijk}  Z^i \ket{+} Z^j \ket{+} \\
\ket{E_{ijk}^1} = \eta_{ijk}  Z^i \ket{+} Z^j \ket{+}.
\end{align*}
The environment states are orthogonal unless $(ij) = (kl)$, thus the set of conditions~\ref{eq:SubsystemConditions1}--\ref{eq:SubsystemConditions2} will be trivially satisfied unless $(ij) = (kl)$. Therefore let~$(ij) = (kl)$, Equation~\ref{eq:SubsystemConditions1} then becomes
\begin{align*}
|\gamma_{ij0}|^2 + |\gamma_{ij1}|^2 = |\eta_{ij0}|^2 + |\eta_{ij1}|^2.
\end{align*}
Each side of the above equation will have one non-zero term that will be equal to 1 as all the coefficients in the encoding are of equal magnitude, therefore Equation~\ref{eq:SubsystemConditions1} will always be satisfied. The condition set out by Equation~\ref{eq:SubsystemConditions2} will have the following form when $(ij)=(kl)$, 
\begin{align*}
\gamma_{ij0}^* \eta_{ij0} = -\gamma_{ij1}^* \eta_{ij1},
\end{align*}
yet since the logical states have support on differing computational basis states, both sides of the above equation will always be equal to zero as for any $(ijm)$, $\gamma_{ijm}^* \eta_{ijm} = 0$.

\end{appendix}

\end{document}